\documentclass{theoretics}

\title{Optimal Padded Decomposition For Bounded Treewidth Graphs}

\ThCSauthor[BarIlan]{Arnold~Filtser}{arnold.filtser@biu.ac.il}[0000-0001-9578-9304]
\ThCSauthor[HassoPlattner]{Tobias Friedrich}{tobias.friedrich@hpi.de}[0000-0003-0076-6308]
\ThCSauthor[NextSilicon]{Davis Isaac}{davisisaac22@gmail.com}[0000−0001−5559−7471]
\ThCSauthor[Waterloo]{Nikhil Kumar}{nikhil.kumar2@uwaterloo.ca}[0000-0001-8634-6237]
\ThCSauthor[Amherst]{Hung Le}{hungle@cs.umass.edu}[0000-0001-8223-9944]
\ThCSauthor[HassoPlattner]{Nadym Mallek}{nadym.mallek@hpi.de}[0000-0002-4370-5145]
\ThCSauthor[HassoPlattner]{Ziena Zeif}{ziena.zeif@hpi.de}[0000-0003-0378-1458]

\ThCSaffil[BarIlan]{Bar-Ilan University, Ramat Gan, Israel}
\ThCSaffil[HassoPlattner]{Hasso Plattner Institute, University of Potsdam, Germany}
\ThCSaffil[NextSilicon]{NextSilicon, Berlin, Germany}
\ThCSaffil[Waterloo]{University of Waterloo, Waterloo Ontario, Canada}
\ThCSaffil[Amherst]{University of Massachusetts Amherst, USA}

\ThCSthanks{Part of the results in this paper were previously published in conference STOC'23 \cite{FIKMZ23}. \\ Arnold Filtser was supported by the Israel Science Foundation (grant No. 1042/22). Hung Le was supported by the NSF CAREER Award No. CCF-223728, an NSF Grant No. CCF-2121952, and a Google Research Scholar Award.}
\ThCSshortnames{A.\ Filtser, T.\ Friedrich, D.\ Isaac, N.\ Kumar, H.\ Le, N.\ Mallek, Z.\ Zeif }  
\ThCSshorttitle{Optimal Padded Decomposition For Bounded Treewidth Graphs}
\ThCSyear{2025}
\ThCSarticlenum{22}
\ThCSreceived{Nov 11, 2024}
\ThCSrevised{Jul 18, 2025}
\ThCSaccepted{Aug 17, 2025}
\ThCSpublished{Oct 10, 2025}
\ThCSdoicreatedtrue
\ThCSkeywords{padded decompositions, treewidth, tree-ordered nets, sparse covers, padded partition covers, metric embeddings, max-flow/min-multicut}

\allowdisplaybreaks
\usepackage{amsfonts}
\usepackage{thm-restate}
\usepackage{tcolorbox}
\usepackage[capitalise,noabbrev,nameinlink]{cleveref}

\newcommand{\tp}{{\rm tp}}
\newcommand{\tw}{{\rm tw}}
\newcommand{\pw}{{\rm pw}}
\newcommand{\Lip}{{\rm Lip}}
\newcommand{\poly}{\mathrm{poly}}

\newcommand{\R}{\mathbb{R}}
\newcommand{\diam}{\mathrm{diam}}
\newcommand{\supp}{\mathrm{supp}}
\newcommand{\Texp}{\mathsf{Texp}}
\newcommand{\mt}{\mathcal{T}}
\newcommand{\rank}{\mathrm{rank}}
\newcommand{\mP}{\mathcal{P}}
\newcommand{\etal}{{\em et al. \xspace}}
\newcommand{\UST}{\textsf{UST}\xspace}
\newcommand{\OPT}{\text{OPT}}
\newcommand{\UTSP}{\textsf{UTSP}\xspace}
\newcommand{\TSP}{\textsf{TSP}\xspace}
\newcommand{\rt}{\mbox{\rm rt}}
\newcommand{\calD}{\mathcal{D}}
\newcommand{\E}{{\mathbb{E}}}
\newcommand{\ProblemName}[1]{\textsf{#1}}
\newcommand{\ZEX}{\ProblemName{$0$-Extension}\xspace}

\def\cF{\ensuremath{\mathcal{F}}}
\def\coreSet{\mathcal{R}}
\def\uncovset{\mathtt{Uncov}}
\def\cP{\ensuremath{\mathcal{P}}}
\def\ball{\mathbf{B}}
\def\eps{\varepsilon}
\def\cC{\ensuremath{\mathcal{C}}}

\def\core{R}
\def\shadow{shadow}
\def\centre{Y}
\def\C{\mathfrak{C}}
\def\barC{\bar{\C}}

\usepackage{todonotes}

\crefname{conjecture}{Conjecture}{Conjectures}
\crefname{observation}{Observation}{Observations}
\crefname{claim}{Claim}{Claims}

\makeatletter
\newcommand\IfRestateTF{%
  \ifx\label\thmt@gobble@label 
    \expandafter\@firstoftwo
  \else
    \expandafter\@secondoftwo
  \fi
}
\makeatother
\newcommand{\RestateRemark}{\IfRestateTF{{\normalfont\bfseries (Restated) }}{}}

\begin{document}
\maketitle

\begin{abstract}
  A $(\beta,\delta,\Delta)$-padded decomposition of an edge-weighted graph $G = (V,E,w)$ is a stochastic decomposition into clusters of diameter at most $\Delta$ such that for every vertex $v\in V$, the probability that $\ball_G(v,\gamma\Delta)$ is entirely contained in the cluster containing $v$ is at least $e^{-\beta\gamma}$ for every $\gamma \in [0,\delta]$. Padded decompositions have been studied for decades and have found numerous applications, including metric embedding, multicommodity flow-cut gap, multicut, and zero extension problems, to name a few. In these applications, parameter $\beta$, called the \emph{padding parameter}, is the most important parameter since it decides either the distortion or the approximation ratios. For general graphs with $n$ vertices, the \emph{padding parameter} $\beta$ is known to be in $\Theta(\log n)$.
 
 Klein, Plotkin, and Rao~\cite{KPR93} (KPR) showed that $K_r$-minor-free graphs have padding parameter $\beta = O(r^3)$, which is a significant improvement over general graphs when $r$ is a constant. However, when $r= \Omega(\log n)$, the padding parameter in KPR decomposition can be much worse than $\log n$.  A long-standing conjecture is that constructing a padded
decomposition for $K_r$-minor-free graphs is possible with padding parameter $\beta = O(\log r)$. Despite decades of research, the best-known result is $\beta = O(r)$, even for graphs with treewidth at most $r$.

In this work, we make significant progress toward the aforementioned conjecture by showing that graphs with treewidth $\tw$ admit a padded decomposition with padding parameter $O(\log \tw)$, which is tight. Our padding parameter is strictly better than $O(\log n)$ whenever $\tw = n^{o(1)}$, and is never worse than what is known for general graphs. As corollaries, we obtain an exponential improvement in dependency on treewidth in a host of algorithmic applications: $O(\sqrt{ \log n \cdot \log(\tw)})$ flow-cut gap,  the maxflow-min multicut ratio of $O(\log(\tw))$, an $O(\log(\tw))$ approximation for the 0-extension problem, an $\ell^{O(\log n)}_\infty$ embedding with distortion $O(\log \tw)$, and an $O(\log \tw)$ bound for integrality gap for the uniform sparsest cut.
\end{abstract}



\section{Introduction} \label{sec:intro}
A basic primitive in designing divide-and-conquer graph algorithms is partitioning a graph into clusters such that there are only a few edges between clusters. This type of primitive has been extensively applied in the design of divide-and-conquer algorithms. Since guaranteeing a few edges crossing different clusters in the worst case could be very expensive, we often seek a good guarantee in a probabilistic sense: the probability that two vertices $u$ and $v$ are placed into two different clusters is proportional to $d_G(u,v)/\Delta$ where $d_G(u,v)$ is the distance between $u$ and $v$ in the input graph $G$ and $\Delta$ is the upper bound on the diameter of each cluster. A (stochastic) partition of $V(G)$ with this property is called a \emph{separating decomposition} of $G$~\cite{Fil19Approx}.

In this work, we study a stronger notion of stochastic decomposition, called \emph{padded decomposition}. More formally, given a weighted graph $G=(V,E,w)$, a partition is $\Delta$-\emph{bounded} if the diameter of every cluster is at most $\Delta$. A distribution $\mathcal{D}$ over partitions is called a $(\beta,\delta,\Delta)$-\emph{padded decomposition}, if every partition in the support is $\Delta$-bounded, and for every vertex $v\in V$ and $\gamma\in[0,\delta]$, we have:

\begin{equation}
	\Pr[\ball_G(v,\gamma\Delta)\subseteq P(v)] \ge e^{-\beta\gamma}  \qquad \text{where $P(v)$ is the cluster containing $v$.}
\end{equation}

That is, the probability that the entire ball $\ball_G(v,\gamma\Delta)$ of radius $\gamma\Delta$ around $v$ is clustered together, is at least $e^{-\beta\gamma}$. If $G$ admits a $(\beta,\delta,\Delta)$-padded decomposition for every $\Delta>0$, we say that $G$ admits $(\beta,\delta)$-padded decomposition scheme. The parameter $\beta$ is usually referred to as a padding parameter.

In an influential work, Klein, Plotkin and Rao \cite{KPR93} showed that every $K_r$ minor free graph admits a weak $\left(O(r^3),\Omega(1)\right)$-padded decomposition scheme; the padding parameter is $O(r^3)$.  This result has found numerous algorithmic applications for solving problems in $K_r$-minor-free graphs; a few examples are the flow-cut gap of $O(r^3)$ for uniform multicommodity flow~\cite{KPR93}, extending Lipschitz functions with absolute extendability of $O(r^3)$~\cite{LN05},  the maxflow-min multicut ratio of $O(r^3)$ for the multicommodity flow with maximum total commodities~\cite{TV93}, an $O(r^3 \log\log(n))$ approximation for minimum linear arrangement,  minimum containing interval graphs~\cite{RR05}, an $O(r^3)$ approximation for the 0-extension problem~\cite{CKR04}, and an $O(r^3\log n)$-approximation for the minimum bisection problem~\cite{FK02}. An important takeaway is that key parameters quantifying the quality of these applications depend (linearly) on the padding parameter; any improvement to the padding parameter would imply the same improvement in the applications. 

Fakcharoenphol and Talwar \cite{FT03} improved the padding parameter of $K_r$ minor free graphs to $O(r^2)$.  Abraham, Gavoille, Gupta, Neiman, and Talwar \cite{AGGNT19} (see also \cite{Fil19Approx}) improved the padding parameter to $O(r)$. These improvements imply an $O(r)$ dependency on the minor size of all aforementioned applications. The only lower bound is $\Omega(\log r)$ coming from the fact that $r$-vertex expanders (trivially) exclude $K_r$ as a minor while having padding parameter $\Omega(\log r)$ \cite{Bar96}. Closing the gap between the upper bound of $O(r)$ and the lower bound $\Omega(\log r)$ has been an outstanding problem asked by various authors \cite{FT03,Lee12,AGGNT19,Fil19Approx}.

\begin{conjecture}\label{conj:Lee} There exists a padded decomposition of any $K_r$-minor-free metric with padding parameter $\beta = O(\log r)$.
\end{conjecture}

Progress on \Cref{conj:Lee} has been made on very special classes of minor-free graphs. Specifically, graphs with \emph{pathwidth} $\pw$ admit padding parameter $O(\log\pw)$ \cite{AGGNT19}. This result implies that $n$-vertex graphs with treewidth $\tw$ admit padding parameter $O(\log\tw+\log\log n)$, since the pathwidth of graphs of treewidth $\tw$ is $O(\tw\cdot \log(n))$ (also see~\cite{KK17}). However, the padding parameter depends on $n$. Thus, a significant step towards \Cref{conj:Lee} is to show that graphs of treewidth $\tw$ admit a padded decomposition with padding parameter $O(\log\tw)$. The best-known result (without the dependency on $n$) for small treewidth graphs is the same as minor-free graphs, implied by the fact that such graphs of treewidth $\tw$ exclude $K_{\tw+2}$ as a minor.  Our first main result is to prove \Cref{conj:Lee} for the special case of treewidth-$\tw$ graphs:

\begin{theorem}\label{thm:paddedTW}
	Every weighted graph $G$ with treewidth $\tw$ admits a  $\left(O( \log\tw),\Omega(1)\right)$-padded decomposition scheme. Furthermore, such a partition can be sampled efficiently.
\end{theorem}

Our \Cref{thm:paddedTW} implies that we could replace $r^3$ with $O(\log \tw)$ in the aforementioned problems when the input graphs have treewidth $\tw$: $O(\log(\tw))$ for uniform multicommodity flow-cut gap, extending Lipschitz functions with absolute extendability of $O(\log(\tw))$,  the maxflow-min multicut ratio of $O(\log(\tw))$, an $O(\log(\tw) \log\log(n))$ approximation for minimum linear arrangement,  minimum containing interval graphs, an $O(\log(\tw))$ approximation for the 0-extension problem,  and an $O(\log(\tw)\log n)$-approximation for the minimum bisection problem. Furthermore, we obtain  the first $\ell_1$ embedding of treewidth-$\tw$ metrics with distortion $\mathcal{O}(\sqrt{\log \tw \cdot \log n})$, an $\ell^{O(\log n)}_\infty$ embedding with distortion $O(\log \tw)$, and $O(\log \tw)$ bound for integrality gap for the uniform sparsest cut.  For several of these problems, for example, uniform multicommodity flow-cut gap, maxflow-min multicut ratio, and 0-extension, our results provide the state-of-the-art approximation ratio for an entire range of parameter $\tw$, even when $\tw = \Omega(n)$.  We refer readers to \Cref{sec:apps} for a more comprehensive discussion of these results. 

Here we point out another connection to \Cref{conj:Lee}. Filtser and Le~\cite{FL22} showed that one can embed any $K_r$-minor-free metric of diameter $\Delta$ into a graph with treewidth $O_r(\epsilon^{-2}\cdot(\log\log n)^2)$ and additive distortion $\epsilon\cdot \Delta$. The dependency of $O_r(\cdot)$ on $r$ is currently huge; it is the constant in the Robertson-Seymour decomposition. However, if one could manage to get the same treewidth to be $O(\poly(r/\eps))$, then in combination with our \Cref{thm:paddedTW}, one has a positive answer to \Cref{conj:Lee}. Even an embedding with a treewidth $O(\poly(r/\eps)\poly(\log(n))$ already implies a padding parameter $O(\log(r) + \log\log(n))$, a significant progress towards \Cref{conj:Lee}. 

\paragraph*{Sparse Covers.~}
A related notion to padded decompositions is \emph{sparse cover}.
A collection $\mathcal{C}$ of clusters is a $(\beta,s,\Delta)$-sparse cover if it is $\Delta$-bounded, each ball of radius $\frac\Delta\beta$ is contained in some cluster, and each vertex belongs to at most $s$ different clusters. 
A graph admits $(\beta,s)$-sparse cover scheme if it admits $(\beta,s,\Delta)$-sparse cover for every $\Delta>0$.

Sparse covers have been studied for various classes of graphs, such as general graphs \cite{AP90}, planar graphs~\cite{BLT14}, minor-free graphs~\cite{KLMN04,BLT14,AGMW10},  and doubling metrics~\cite{Fil19Approx}. 

By simply taking the union of many independently drawn copies of padded decomposition, one can construct a sparse cover. Indeed, given $(\beta,\delta,\Delta)$-padded decomposition, by taking the union of $O(e^{\beta\gamma}\log n)$ partitions (for $\gamma\le\delta$) one will obtain w.h.p. a $(\gamma,O(e^{\beta\gamma}\log n),\Delta)$-sparse cover. In particular, using \Cref{thm:paddedTW} one can construct $(O(1),O(e^{\log\tw}\log n))=(O(1),\tw^{O(1)}\log n))$-sparse cover scheme. The main question in this context is whether one can construct sparse covers for bounded treewidth graphs with constant cover parameter ($\beta$) and sparseness ($s$) independent from $n$. If one is willing to sacrifice a (quadratic) dependency on $\tw$ on the cover parameter $\beta$, then a sparse cover with parameters independent of $n$ is known~\cite{KPR93,FT03,AGGM06}. However, in many applications, for example, constructing sparse spanners, it is desirable to have $\beta = O(1)$ as it directly governs the stretch of the spanners.

\begin{restatable}{theorem}{CoverTWTheorem}\label{thm:CoverTW} \RestateRemark Every graph $G$ with treewidth $\tw$ admits a  $\left(6,\poly(\tw)\right)$-sparse cover scheme. 
\end{restatable}

It is sometimes useful to represent the sparse cover $\cC$ as a union of partitions, for example, in 
metric embeddings \cite{KLMN04}, and the construction of ultrametric covers \cite{FL22,Fil23},
leading to the notion of \emph{padded partition cover scheme}:
\begin{definition}[Padded Partition Cover Scheme]\label{def:PaddedPartitionCover} A collection of partitions $\mathcal{P}_{1},\dots,\mathcal{P}_{\tau}$
	is $(\beta,s,\Delta)$-padded partition cover if (a) $\tau\le s$, (b) every partition $\mathcal{P}_{i}$ is $\Delta$-bounded, and (c) for every point $x$, there is a cluster $C$ in one of the partitions $\mathcal{P}_{i}$ such that $B(x,\frac{\Delta}{\beta})\subseteq C$.\\
	A space $(X,d_{X})$ admits a $(\beta,s)$-\emph{padded partition cover scheme} if for every $\Delta$, it admits a  $(\beta,s,\Delta)$-padded  partition cover.
\end{definition}

While a padded partition cover implies a sparse cover with the same parameters, the reverse direction is not true. For example, graphs with pathwidth $\pw$ admit $(10,5(\pw+1))$-sparse cover scheme \cite{Fil20}, however, they are only known to admit  $(O(\pw^2),2^{\pw+1})$-padded partition cover scheme (this is due to $K_r$-minor free graphs \cite{KPR93,FT03} (see also \cite{KLMN04,Fil20})).\footnote{Note that we do not have matching lower bounds, so there is no provable separation.} That is, the sparseness parameter in the padded partition cover scheme is exponentially worse (in terms of $\pw$) than that of the sparse cover scheme. In this work, we construct a padded partition cover scheme with the same quality as our sparse covers.

\begin{theorem}\label{thm:PaddedPartitionCoverTW}
	Every graph $G$ with treewidth $\tw$ admits a  $\left(12,\poly(\tw)\right)$-padded partition cover scheme. 
\end{theorem}

\paragraph*{Tree-Ordered Net.} A key new technical insight to all of our aforementioned results is the notion of \emph{tree-ordered net} (\Cref{def:TreeOrderNet}). A tree order net is analogous to the notion of \emph{nets}. \emph{Nets} were used extensively in designing algorithms for metric spaces. More formally, a $\Delta$-net is a set of points $N$ such that every two net points are at a distance at least $\Delta$ (i.e~$\min_{x,y\in N}d_X(x,y)\ge\Delta$), and every point has a net point at a distance at most $\Delta$ (i.e.~$\max_{x\in V}\min_{y\in N}d_X(x,y)\le\Delta$). Filtser \cite{Fil19Approx} showed that if there is a $\Delta$-net such that every ball of radius $3\Delta$ contains at most $\tau$ net points, then the metric admits a $\left(O(\log\tau),\Omega(1),O(\Delta)\right)$-padded decomposition. This result implies that metrics of doubling dimension $d$  have padding parameter $O(d)$ since doubling metrics have sparse nets: $\tau = 2^{O(d)}$. Unfortunately, graphs of small treewidth do not have sparse nets; this holds even in very simple graphs such as star graphs. Nonetheless, we show that small treewidth graphs possess a structure almost as good: a net that is sparse w.r.t. some partial order. We formalize this property via tree-ordered nets.  As we will later show, a sparse tree-ordered net is enough to construct the padded decomposition scheme; see \Cref{sec:centersToPadded}. We believe that the notion of a tree-ordered net is of independent interest.

\begin{figure}
\label{fig:treeOrder}
		\includegraphics[width=0.4\textwidth]{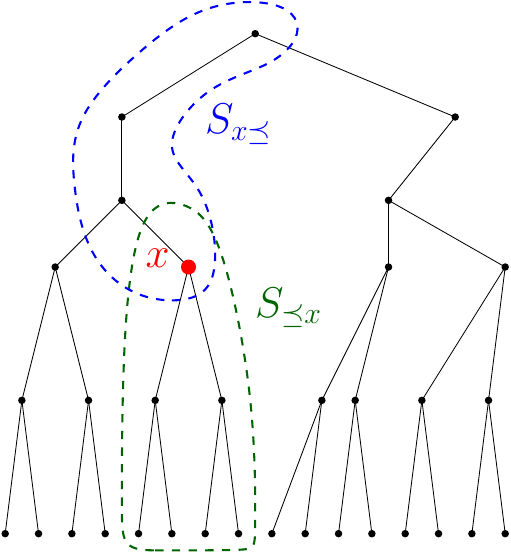}\caption{Illustration of a tree order.}
\end{figure}
\paragraph*{Tree-Order.} Another important insight, is that of tree orders. A tree order $\preceq$ of a set $V$ is a partial order (i.e. transitive, reflexive, and antisymmetric) associated with a rooted tree $T$ and a map $\varphi: V\rightarrow V(T)$ such that $u\preceq v$ iff $\varphi(v)$ is an ancestor of $\varphi(u)$ in $T$. Given a weighted graph $G=(V,E,w)$, a tree order  $\preceq$ w.r.t. tree $T$ is a \emph{valid order} of $G$ if for every edge $\{u,v\}\in E$, it holds that  $u\preceq v$ or $v\preceq u$ or both (i.e. $v$ is an ancestor of $u$ or $u$ is an ancestor of $v$).  A simple consequence of the validity is that every connected subset $C$ in $G$ must contain a maximum element w.r.t $\preceq$; see \Cref{obs:maximal}.
For a vertex $v\in V$, and subset $S\subseteq V$ let $S_{v\preceq}=\{u\in S\mid v\preceq u\}$ be the ancestors of $v$ w.r.t. $T$ in $S$. Similarly, let $S_{\preceq x}=\{u\in S\mid u\preceq x\}$ be all the descendants of $x$ in $S$.

\begin{definition}\label{def:TreeOrderNet}
	Given a weighted graph $G=(V,E,w)$ and parameters $\tau,\alpha, \Delta > 0$, a $(\tau,\alpha,\Delta)$-tree-ordered net is a triple $(N,T,\varphi)$ where $N\subseteq V$, and $T$ and $\varphi$ define a tree order $\preceq$ of $V$ such that for every $v\in V$:
    \begin{itemize}[topsep=0pt]

		\item \textsc{Covering.} There is $x\in N_{v\preceq}$ such that $d_{G[V_{\preceq x}]}(v,x)\le \Delta$. That is, there exists an ancestor $x$ of $v$ in $N$ such that the distance from $v$ to $x$ in the subgraph of $G$ induced by descendants of $x$ is at most $\Delta$.
		\item \sloppy \textsc{Packing.}  Denote by $N_{v\preceq}^{\alpha\Delta}=\left\{ x\in N_{v\preceq}\mid d_{G[V_{\preceq x}]}(v,x)\le\alpha\Delta\right\}$ the set of ancestor centers of $v$ at distance at most $\alpha\Delta$ from $v$ (w.r.t the subgraphs induced by descendants of the ancestors. Then $\left|N_{v\preceq}^{\alpha\Delta}\right|\le\tau$.
	\end{itemize}	
\end{definition}

While we state our main result in terms of treewidth, it will be more convenient to use the notion of  \emph{bounded tree-partition width} \cite{DBLP:journals/dm/DingO96}. We will show that graphs of bounded tree-partition width admit a small tree-ordered net.

\begin{definition}[Tree Partition] A \emph{tree partition} of a graph  $G=(V,E)$ is a  rooted tree $\mathcal{T}$ whose vertices are bijectively associated with the sets of  \emph{partition} $\mathcal{S} = \{S_1,S_2,\ldots,S_m\}$ of $V$, called \emph{bags}, such that for each $(u,v) \in E$, there exists a $S_i,S_j \in \mathcal{S}$ such that $S_j$ is the parent of $S_i$ and $\{u,v\} \subseteq S_i \cup S_j$. The width of $\mathcal{T}$ is $\max_{i\in m}\{|S_i|\}$.
\end{definition}

Unlike a tree decomposition, bags of a tree partition are disjoint.  Therefore, graphs of bounded-tree partition width have a more restricted structure than graphs of bounded treewidth. Indeed, one can show that any graphs of tree-partition width $k$ have treewidth at most $2k-1$. However, from a metric point of view, graphs of bounded treewidth are the same as graphs of bounded tree-partition width: We could convert a tree decomposition into a tree partition by making copies of vertices; see \Cref{lm:TP-from-TD}. We will show in \Cref{sec:treenet-tp} that:

\begin{restatable}{lemma}{NetMainLemma} \label{lem:NetMainLemma}	\RestateRemark
	Every weighted graph  $G=(V,E,w)$ with a tree-partition width $\tp$ admits a $(\poly(\tp),3,\Delta)$-tree-ordered net, for every $\Delta>0$.
\end{restatable}

In \Cref{sec:centersToPadded}, we show how to use the tree-ordered net in \Cref{lem:NetMainLemma} to obtain all results stated above.

\paragraph*{Follow-up Work.~} Recently, inspired by our technique, \cite{CCLMST23} constructed a tree cover with stretch $1+\eps$ and $2^{(\tw / \eps)^{O(\tw)}}$ trees. This result has applications to constructing distance oracles and approximate labeling schemes for graphs of small treewidth. Filtser~\cite{Filtser24} showed recently that $K_r$-minor-free graphs admit a $(4+\eps,O(1/\eps)^{r})$-sparse cover scheme for every  $\eps \in (0,1)$. While the stretch is an absolute constant, the dependency on the minor size is exponential and hence is incomparable to our \Cref{thm:CoverTW}.

In a subsequent work \cite{CF25}, Conroy and Filtser answered the main open question left by this paper (see \Cref{subsec:stocdecomp}), and showed that every $K_{r}$-minor-free graph admits a padded decomposition with padding parameter $O(\log r)$. Conroy and Filtser have been inspired by the techniques developed in the current paper. Specifically, they use the buffered cup decomposition of \cite{CCL+24} and, in the spirit of the current paper, recursively remove many supernodes at once to reduce the treewidth of each remaining component. Conroy and Filtser establish their result by showing the existence of good sparse covers for $K_{r}$-minor-free graphs and a general reduction from sparse cover to padded decomposition.

\section{Preliminaries}
\paragraph*{Graphs.}
We consider connected undirected graphs $G=(V,E)$ with edge weights $w: E \to \R_{\ge 0}$. We say that vertices $v$ and $u$ are neighbors if $\{v,u\}\in E$. Let $d_{G}$ denote the shortest path metric in $G$. $\ball_G(v,r)=\{u\in V\mid d_G(v,u)\le r\}$ is the ball of radius $r$ around $v$. For a vertex $v\in V$ and a subset $A\subseteq V$, let $d_{G(x,A)}:=\min_{a\in A}d_G(x,a)$, where $d_{G}(x,\emptyset)= \infty$. For a subset of vertices $A\subseteq V$, let $G[A]$ denote the induced graph on $A$, and let $G\setminus A := G[V\setminus A]$.

The \emph{diameter} of a graph $G$ is $\diam(G)=\max_{v,u\in V}d_G(v,u)$, i.e. the maximal distance between a pair of vertices.
Given a subset $A\subseteq V$, the \emph{weak}-diameter of $A$ is $\diam_G(A)=\max_{v,u\in A}d_G(v,u)$, i.e. the maximal distance between a pair of vertices in $A$, w.r.t. to original distances $d_G$. The \emph{strong}-diameter of $A$ is $\diam(G[A])$, the diameter of the graph induced by $A$. 
A graph $H$ is a \emph{minor} of a graph $G$ if we can obtain $H$ from
$G$ by edge deletions/contractions, and isolated vertex deletions.  A graph
family $\mathcal{G}$ is \emph{$H$-minor-free} if no graph
$G\in\mathcal{G}$ has $H$ as a minor. We will drop the prefix $H$ in $H$-minor-free whenever $H$ is not important or clear from the context. 

Some examples of minor-free graphs are planar graphs ($K_5$- and $K_{3,3}$-minor-free), outer-planar graphs ($K_4$- and $K_{3,2}$-minor-free), series-parallel graphs ($K_4$-minor-free) and trees ($K_3$-minor-free).

\paragraph*{Treewidth.} A \emph{tree decomposition} of a graph $G = (V,E)$ is a  tree $\mathcal{T}$ where each node $x\in \mathcal{T}$ is associated with a subset $S_x$ of $V$, called a \emph{bag}, such that: (i) $\cup_{x\in V(\mathcal{T})} S_x = V$, (ii) for every edge $(u,v) \in E$, there exists a bag $S_x$ for some $x\in V(\mathcal{T})$ such that $\{u,v\}\subseteq S$, and (iii) for every $u\in V$, the bags containing $u$ induces a connected subtree of $\mathcal{T}$. The \emph{width} of $\mathcal{T}$ is $\max_{x\in V(\mathcal{T})}\{|S_x|\}$-1. The \emph{treewidth} of $G$ is the minimum width among all possible tree decompositions of $G$.

\paragraph*{Padded Decompositions.}
Consider a \emph{partition} $\mathcal{P}$ of $V$ into disjoint clusters.
For $v\in V$, we denote by $P(v)$ the cluster $P\in \mathcal{P}$ that contains $v$.
A partition $\mathcal{P}$ is strongly $\Delta$-\emph{bounded} (resp. weakly $\Delta$-bounded ) if the strong-diameter (resp. weak-diameter) of every $P\in\mathcal{P}$ is bounded by $\Delta$.
If the ball $\ball_G(v,\gamma\Delta)$ of radius $\gamma\Delta$ around a vertex $v$ is fully contained in $P(v)$, we say that $v$ is $\gamma$-{\em padded} by $\mathcal{P}$. Otherwise, if $\ball_G(v,\gamma\Delta)\not\subseteq P(v)$, we say that the ball is \emph{cut} by the partition.

\begin{definition}[Padded Decomposition]\label{def:PadDecompostion}
	Consider a weighted graph $G=\left(V,E,w\right)$.
	A distribution $\mathcal{D}$ over partitions of $G$ is strongly (resp. weakly) $(\beta,\delta,\Delta)$-padded decomposition if every $\mathcal{P}\in\supp(\mathcal{D})$ is strongly (resp. weakly) $\Delta$-bounded and for any $0\le\gamma\le\delta$, and $z\in V$,
	$$\Pr[\ball_G(z,\gamma\Delta)\subseteq P(z)] \ge e^{-\beta\gamma}~.$$
	We say that a graph $G$ admits a strong (resp. weak) $(\beta,\delta)$-padded decomposition scheme, if for every parameter $\Delta>0$ it admits a strongly (resp. weakly) $(\beta,\delta,\Delta)$-padded decomposition that can be sampled in polynomial time.  
\end{definition}

\subsection{From Tree Decomposition to Tree Partition} \label{sec:tree-partition}

We first show that any graph of treewidth $\tw$ can be embedded isometrically into a graph of tree-partition width $\tw$ by duplicating vertices. 

\begin{figure}[!ht]
	\begin{center}
		\includegraphics[scale=0.9]{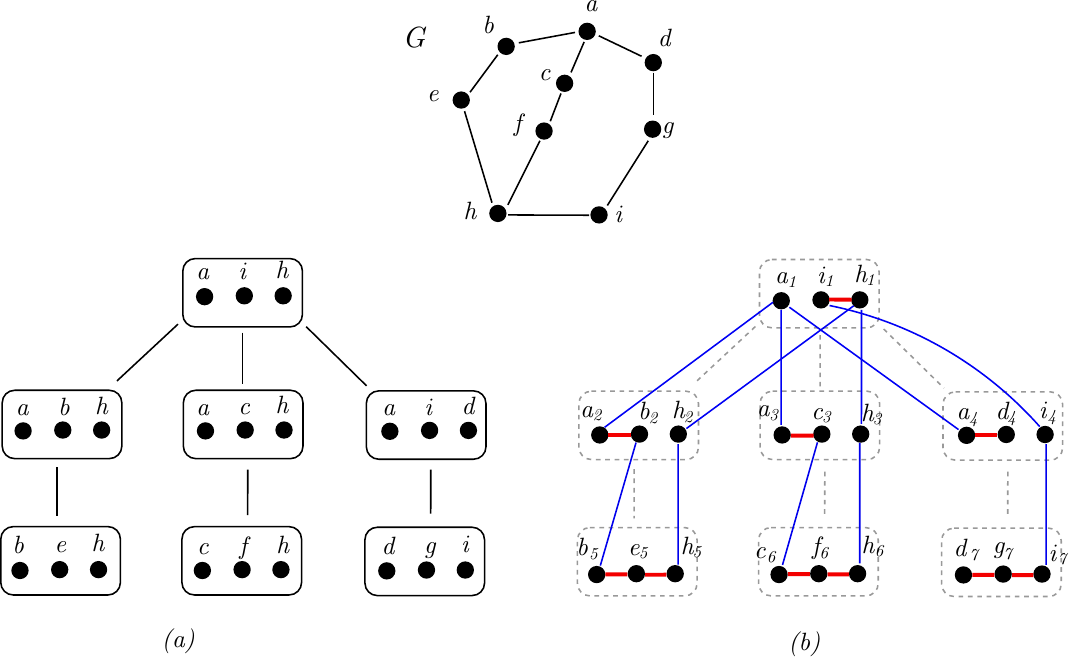} 
	\end{center}
	\caption{Converting a tree decomposition in (a) into a tree partition in (b). The thick red edges are the original edges in $G$, and the thin blue edges are the added edges with weight zero.}
	\label{fig:treeToWidth}
\end{figure}

\begin{lemma}\label{lm:TP-from-TD} 
	Given an edge-weighed graph $G(V,E,w)$ and its tree decomposition of width $\tw$, there is an isometric polynomial time constructible embedding of $G$ into a graph with tree partition of width $\tw+1$. More formally, there is a graph $H=(X,E_H,w_H)$ with tree partition of width $\tw+1$ and a map $\phi: V\rightarrow X$ such that $\forall x,y\in V$, $d_H(\phi(x),\phi(y))=d_G(x,y)$.
\end{lemma}
\begin{proof} Let $\mathcal{B}$ be a tree decomposition of width at most $r$ of $G = (V,E,w)$.  We create a graph $H$ and its tree partition as follows (see \Cref{fig:treeToWidth}). For each $u \in V$, if $u$ appears in $k$ bags of $\mathcal{B}$, say  $B_1,B_2,\ldots,B_k$, then we make $k$ copies of $u$, say $u_1,u_2,\ldots, u_k$ and replace $u$ in bag $B_i$ with its copies $u_i$, $i\in [1,k]$.  This defines the set of vertices $X$ of $H$. We then set $\phi(u) = u_1$. 
	
	If $B_i$ and $B_j$ are two adjacent bags in $\mathcal{B}$, we create an edge $(u_i,u_j)$ and assign a weight $w_H(u_i,u_j) = 0$.  For each $(u,v) \in E$, there exists at least one bag $B_j$ of the tree decomposition of $G$ such that $u,v \in B_j$. We add the edge $(u_j,v_j)$ of weight $w_H(u_j,v_j) = w(u,v)$ to $E_H$. This completes the construction of $H$. 
	
	The tree partition of $H$, say $\mathcal{T}$, has the same structure as the tree decomposition $\mathcal{B}$ of $G$: for each bag $B\in \mathcal{B}$, there is a corresponding bag $\hat{B} \in \mathcal{T}$ containing copies of vertices of $B$. As $|B|\leq r+1$, the width of $\mathcal{T}$ is $r+1$.  For each $u\in V$, we map it to exactly one copy of $u$ in $X$. As edges between copies of the same vertex have 0 weight, the distances in $G$ are preserved exactly in $H$.
\end{proof}

Using \Cref{lm:TP-from-TD} together with the tree-ordered net in \Cref{lem:NetMainLemma} we will construct (weak) padded decomposition and sparse covers for graphs of small treewidth. 

\section{Padded Decomposition and Sparse Cover}\label{sec:centersToPadded}

In this section, we prove three general lemmas on constructing a padded decomposition and a sparse cover from a tree-ordered net. These lemmas, together with \Cref{lem:NetMainLemma}, imply \Cref{thm:paddedTW}, \Cref{thm:CoverTW}, and \Cref{thm:PaddedPartitionCoverTW}. We begin by proving \Cref{thm:CoverTW}, as its proof is the simplest among the three. We observe the following by the definition of a tree order.

\begin{observation}\label{obs:maximal} Let $\preceq$ be a valid tree order of $G = (V,E,w)$ defined by a tree $T$. Every subset of vertices $C$ such that $G[C]$ is connected must contain a single maximum element w.r.t $\preceq$.
\end{observation}
\begin{proof}
	Suppose that there are $t$ maximal elements $u_1,u_2,\ldots, u_t\in C$ for $t\geq 2$. Let  $A_i = \{v: v\preceq u_i\}$ for every $i\in [t]$ and  $\mathcal{A} = \{A_1,\ldots, A_t\}$. Then $\mathcal{A}$ is a partition of $C$ since $\{u_i\}_{i=1}^t$ are maximal (and $T$ is a tree). As $G[C]$ is connected and $t\geq 2$, there must be some edge $\{x,y\}\in E$ such that $x\in A_i$ and $y\in A_j$ for $i\not= j$. The validity of $T$ implies that either $x\preceq y$ or $y\preceq x$. However, this means either $x$ is also in $A_j$ or $y$ is also in $A_i$, contradicting that $\mathcal{A}$ is a partition of $C$. 
\end{proof}

\subsection{Proof of Theorem~\ref{thm:CoverTW}}\label{subsec:covertw}
The following lemma is a reduction from tree-ordered nets to sparse covers.
\begin{lemma}\label{lem:NetToCover}
	Consider a weighted graph $G=(V,E,w)$ with a $(\tau,\alpha,\Delta)$-tree-ordered net (w.r.t. a tree order $\preceq$, associated with a tree $T$), then $G$ admits a  strong 
	 $\left(\frac{4\alpha}{\alpha-1},\tau,2\alpha\Delta\right)$-sparse cover that can be computed efficiently.
\end{lemma}
\begin{proof}Let $N$ be the $(\tau,\alpha,\Delta)$-tree-ordered net of $G$.  Let $x_1,x_2,\dots$ be an ordering of the centers in $N$ w.r.t distance from the root in $T$. Specifically, $x_1$ is closest to the root, and so on. Note that distances in $T$ are unweighted and unrelated to $d_G$; we break ties arbitrarily. For every vertex $x_i\in N$, create a cluster
	\[
	C_{i}=\ball_{G[V_{\preceq x_{i}}]}(x_i,\alpha\Delta)~,
	\]
	of all the vertices that are at distances at most $\alpha\Delta$ from $x_i$ in the subgraph induced by the descendants of $x_i$. We now show that $\mathcal{C} = \{C_i\}_i$ is the sparse cover claimed in the lemma.

	Clearly, by the triangle inequality, every cluster in $\mathcal{C}$ has a diameter at most $2\alpha\Delta$. Furthermore, observe that every vertex $v$ belongs to at most $\tau$ clusters since it has at most $\tau$ ancestors in $N$ at distance at most $\alpha\Delta$ (in the respective induced graphs).
	
	Finally, we show that for every vertex $v$, $\ball_G(v,\frac\Delta\beta)$ is fully contained in some cluster in $\mathcal{C}$, for $\beta=\frac{2}{\alpha-1}$. Let $B=\ball_G(v,\frac{\Delta}{\beta})$.	Let $v_B\in B$ be the closest vertex to the root w.r.t. $T$. \Cref{obs:maximal} implies that for every vertex $u\in B$,  $u\preceq v_B$. By the triangle inequality, for every $u\in B$ it holds that:
	\begin{equation}\label{eq:u-vB}
		d_{G[V_{\preceq v_B}]}(v_B,u)\le d_{G[V_{\preceq v_B}]}(v_B,v)+d_{G[V_{\preceq v_B}]}(v,u)\le\frac{2\Delta}{\beta}
	\end{equation}
	Let $x_B\in N$ be the ancestor of $v_B$ in $T$ that minimizes $d_{G[V_{\preceq x}]}(x,v_B)$. Since $N$ is a $(\tau,\alpha,\Delta)$-tree-ordered net, it holds that $d_{G[V_{\preceq x_B}]}(x_B,v_B)\le\Delta$. In particular, for every $u\in B$ it holds that:
	\begin{equation*}
		\begin{split}
			d_{G[V_{\preceq x_B}]}(x_B,u) &\le d_{G[V_{\preceq x_B}]}(x_B,v_B)+d_{G[V_{\preceq x_B}]}(v_B,u)\\
			&\le\Delta+\frac{2\Delta}{\beta}\le\alpha\Delta    \qquad \mbox{(by \cref{eq:u-vB})}         
		\end{split}
	\end{equation*}
	Thus, $B$ is fully contained in the cluster centered in $x_B$.
	As the diameter of each cluster is $2\alpha\Delta$, the padding we obtain is $\frac{2\alpha\Delta}{\frac{\Delta}{\beta}}=\frac{2\alpha}{\frac{\alpha-1}{2}}=\frac{4\alpha}{\alpha-1}$, 
	a required.
\end{proof}

Now we are ready to show that \Cref{thm:CoverTW} follows from \Cref{lem:NetMainLemma}, \Cref{lm:TP-from-TD}, and \Cref{lem:NetToCover}. We restate \Cref{thm:CoverTW} for convenience.

\CoverTWTheorem*
\begin{proof}
Let $\Delta > 0$ be a parameter; we will construct a $(6,\poly(\tw),\Delta)$-sparse cover for $G$. Let $H$ be the graph of tree-partition width $\tw+1$ in the isometric embedding $\phi$ of $G$ in \Cref{lm:TP-from-TD}. We abuse notation by using $v$,  for each $v\in V(G)$, to denote $\phi(v)$ in $H$. This means $V(G) \subseteq V(H)$. 

Let $\hat{\Delta} = \Delta/6$. By \Cref{lem:NetMainLemma}, $H$ admits a $(\poly(\tw), 3,\hat{\Delta})$-tree-ordered net. By \Cref{lem:NetToCover}, we can construct a $(6, \poly(\tw), 6\hat{\Delta} = \Delta)$-sparse cover of $H$, denoted by $\hat{\mathcal{C}}$. We then construct $\mathcal{C} = \{\hat{C}\cap V(G): \hat{C}\in \hat{\mathcal{C}}\}$.  Let $C$ be a cluster in $\mathcal{C}$, and we denote by $\hat{C}$ its corresponding cluster in $\hat{\mathcal{C}}$; that is, $C = \hat{C}\cap V(G)$. We claim that $\mathcal{C}$ is  a $(6,\poly(\tw),\Delta)$-sparse cover for $G$.  

For any $C\in \mathcal{C}$, $\diam(C)\leq \diam(\hat{C}) \leq \Delta$, since $\phi$ is an isometric embedding. Thus, $\mathcal{C}$ is $\Delta$-bounded. Furthermore, any $v\in V(G)$ belongs to at most $\poly(\tw)$ clusters in $\mathcal{\hat{C}}$ and hence it belongs to at most  $\poly(\tw)$ clusters  in $\mathcal{C}$. Finally, we consider any ball $B_G(v,\Delta/6)$ in $G$. Observe that $B_G(v,\Delta/6)\subseteq B_H(v,\Delta/6)$. Since $\hat{\mathcal{C}}$ is a $(6, \poly(\tw), \Delta)$-sparse cover of $H$, there exists a cluster $\hat{C} \in \hat{\mathcal{C}}$ such that $B_H(v,\Delta/6)\subseteq \hat{C}$, implying $B_H(v,\Delta/6)\subseteq C$.
\end{proof} 

\subsection{Proof of Theorem~\ref{thm:PaddedPartitionCoverTW}}\label{subsec:PaddedCover}

In this subsection, we show that the construction in \Cref{lem:NetToCover} can be adapted to obtain a padded partition cover scheme (with a small loss in the padding parameter).

\begin{lemma}\label{lem:NetToCoverUnionPartitions}
	Consider a weighted graph $G=(V,E,w)$ with a $(\tau,\alpha,\Delta)$-tree-ordered net (w.r.t. a tree order $\preceq$, associated with $T$) for $\alpha>2$, then 
	$G$ admits strong 
	 $(\frac{4\alpha}{\alpha-2},\tau,\alpha\Delta)$-padded partition cover.	
\end{lemma}

By exactly the same argument in the proof of \Cref{thm:CoverTW} in \Cref{subsec:covertw}, one can show that   \Cref{thm:PaddedPartitionCoverTW} follows directly from \Cref{lem:NetMainLemma}, \Cref{lm:TP-from-TD} and \Cref{lem:NetToCoverUnionPartitions}: first, we isometrically embed $G$ of treewidth $\tw$ to a graph $H$ of tree-partition width $\tw+1$, and then apply \Cref{lem:NetToCoverUnionPartitions} and \Cref{lem:NetMainLemma} to construct a padded partition cover $\hat{\mathcal{P}}$ for $H$, which will then be turned into a padded partition cover for $G$ by removing vertices not in $G$ from $\hat{\mathcal{P}}$.  Our main focus now is to prove \Cref{lem:NetToCoverUnionPartitions}. 

\begin{proof}[Proof of Lemma~\ref{lem:NetToCoverUnionPartitions}]
	
	We say that a set of clusters is a \emph{partial partition} if it is a partition of a subset of vertices; that is, a vertex might not belong to any cluster in a partial partition. We then can turn the partial partitions into partitions of $V$ by adding singleton clusters. Similarly to \Cref{lem:NetToCover}, we define a set of clusters (note radius $\frac\alpha2\cdot\Delta$ compared to $\alpha\cdot\Delta$ in \Cref{lem:NetToCover}):
	\begin{equation*}
		\cC=\left\{C_{x}=\ball_{G[V_{\preceq x}]}(x,\frac\alpha2\cdot\Delta)\right\}_{x\in N}
	\end{equation*}
	Our partial partitions will consist of clusters in $\cC$ only.
	
	We construct the partitions greedily: form a partition from a maximal set of disjoint clusters, preferring ones that are closer to the root, and repeat. The pseudocode is given in \Cref{alg:PaddedClustering}. 
	
	\begin{algorithm}
		\caption{ $\texttt{Create Partial Partitions}((\preceq,T),N,\cC)$}\label{alg:PaddedClustering}
		\DontPrintSemicolon
		
		$A\leftarrow N$\;
		$i\leftarrow0$\;
		\While{$A\ne\emptyset$}{\label{line:all-A}
			$i\leftarrow i+1$\;
			$\cP_i\leftarrow\emptyset$\;
			\While{$\exists x\in A$ such that $C_x$ is disjoint from $\cup_{C\in \cP_i}C$}{
				Let $x\in A$  be a maximal element w.r.t. $(\preceq,T)$ such that $C_x$ is disjoint from $\cup_{C\in \cP_i}C$\label{line:check-disjt}\;
				Remove $x$ from $A$\;
				Add $C_x$ to $\cP_i$\;}
		}
		\Return $\{\cP_j\}_{j=1}^{i}$\;
	\end{algorithm} 

	By construction in \Cref{line:check-disjt}, clusters in every partition $\cP_i$ are pairwise disjoint. Furthermore, for every $x\in N$, $C_x$ belongs to one of the created partitions due to \Cref{line:all-A}.  We next argue that the algorithm creates at most $\tau$ partial partitions.
	
	Suppose for contradiction that the algorithm does not terminate after creating  $\tau$ partial partitions. Thus, after $\tau$ iterations, there was still an element $x\in A$. In particular, in every iteration $i$, there exists some vertex $x_i \in N$ such that $x\preceq x_i$, and $C_x\cap C_{x_i}\not=\emptyset$. 
	Let $y_i$ be  a vertex in $C_x\cap C_{x_i}$.  By the triangle inequality:
	\begin{equation*}
		d_{G}(x,x_{i})\le d_{G[V_{\preceq x}]}(x,y_{i})+d_{G[V_{\preceq x_{i}}]}(y_{i},x_{i})\le\alpha\Delta~.
	\end{equation*}
	Note that as every cluster in $\cC$ can join only one partial partition, all these centers $\{x_i\}_{i=1}^{\tau}$ are unique; this contradicts the fact that $N$ is a $(\tau,\alpha,\Delta)$-tree-ordered net (as together with $x$ itself,~$\left|N_{x\preceq}^{\alpha\Delta}\right|>\tau$). 
	
	It remains to show padding property, which is followed by the same proof as in \Cref{lem:NetToCover}.
	Consider a ball $B=\ball_G(v,\frac{\Delta}{\beta})$, for $\beta=\frac{4}{\alpha-2}$. Let $v_B\in B$ be the closest vertex to the root w.r.t. $T$. Following \cref{eq:u-vB},  for every $u\in B$, $d_{G[V_{\preceq v_B}]}(v_B,u)\le \frac{2\Delta}{\beta}$.
	Let $x_B\in N$ be the ancestor of $v_B$ in $T$ that minimizes $d_{G[V_{\preceq x}]}(x,v_B)$. Since $N$ is a $(\tau,\alpha,\Delta)$-tree-ordered net, it holds that $d_{G[V_{\preceq x_B}]}(x_B,v_B)\le\Delta$. In particular, for every $u\in B$ it holds that 
	\begin{align*}
		d_{G[V_{\preceq x_{B}}]}(x_{B},u) & \le d_{G[V_{\preceq x_{B}}]}(x_{B},v_{B})+d_{G[V_{\preceq x_{B}}]}(v_{B},u)\\
		& \le\Delta+\frac{2\Delta}{\beta}=\left(1+\frac{2\cdot(\alpha-2)}{4}\right)\cdot\Delta=\frac{\alpha}{2}\cdot\Delta~,
	\end{align*}
	and thus $B$ is fully contained in the cluster centered in $x_B$ a required.
\end{proof}

\subsection{Proof of Theorem~\ref{thm:paddedTW}}
In the following lemma, we construct a padded decomposition from a tree-ordered net. 

\begin{lemma}\label{lem:NetToPadded}
	Consider a weighted graph $G=(V,E,w)$ with a $(\tau,\alpha,\Delta)$-tree-ordered net (w.r.t. a tree order $\preceq$, associated with $T$), then $G$ admits a weak $\left(16\cdot\frac{\alpha+1}{\alpha-1}\cdot\ln(2\tau),\frac{\alpha-1}{8\cdot(\alpha+1)},(\alpha+1)\cdot\Delta\right)$-padded decomposition that can be efficiently sampled.
\end{lemma}

By exactly the same argument in the proof of \Cref{thm:CoverTW} in \Cref{subsec:covertw}, one can show that   \Cref{thm:paddedTW} follows directly from \Cref{lem:NetMainLemma}, \Cref{lm:TP-from-TD} and \Cref{lem:NetToPadded}. We will use truncated exponential distribution during the proof of \Cref{lem:NetToPadded}:

\paragraph*{Truncated Exponential Distributions.} To create padded decompositions, similarly to previous works, we will use truncated exponential distributions. A truncated exponential distribution is an exponential distribution conditioned on the event that the outcome lies in a certain interval. More precisely, the \emph{$[\theta_1,\theta_2]$-truncated exponential distribution} with
parameter $\lambda$, denoted by $\Texp_{[\theta_1, \theta_2]}(\lambda)$, has
the density function:
$f(y)= \frac{ \lambda\, e^{-\lambda\cdot y} }{e^{-\lambda \cdot \theta_1} - e^{-\lambda
		\cdot \theta_2}}$, for $y \in [\theta_1, \theta_2]$.

\begin{proof}[Proof of Lemma~\ref{lem:NetToPadded}]	
	Let $x_1,x_2,\dots$ be an ordering of the centers in $N$ w.r.t distances from the root in $T$.
	Set $\beta=\frac{\alpha+1}{2}$.
	For every vertex $x_i\in N$, sample $\delta_i\in[1,\beta]$ according to $\Texp_{[1,\beta]}(\lambda)$, a truncated exponential distribution, with parameter $\lambda=\frac{4}{\alpha-1}\cdot\ln (2\tau)$. Set $R_i=\delta_i\cdot\Delta\in[\Delta,\beta\Delta]$ and create a cluster:
	\[
	C_{i}=\ball_{G[V_{\preceq x_{i}}]}(x_i,R_{i})\setminus\cup_{j<i}C_{j}~.
	\]
	Recall that $\ball_{G[S_{\preceq x_i}]}(x_i,R_{i})$ is the ball of radius $R_i$ around $x_i$ in the graph induced by all the descendants of $x_i$. Thus, the cluster $C_{i}$ of $x_i$ consists of all the points in this ball that did not join the clusters centered at the (proper) ancestors of $x_i$.
	Note that  $C_i$ might not be connected and that $x_i$ might not even belong to $C_i$ as it could join a previously created cluster.  Nonetheless, $C_i$ has a (weak) diameter at most $2R_i \leq 2\beta\Delta=(\alpha+1)\cdot\Delta$ by the triangle inequality.  
	
	We claim that each vertex will eventually be clustered. Indeed, consider a vertex $v\in V$. There exists some vertex $x_i\in N_{v\preceq}$ at a distance at most $\Delta$ from $v$ in $G[V_{\preceq x_i}]$ by the definition of the tree-ordered net.  If $v$ did not join any cluster centered at an ancestor of $x_i$, then $v$ will join $C_i$ because  $d_{G[V_{\preceq x_i}]}(v,x_i)\le \Delta\le R_i$.

	It remains to prove the padding property. Consider some vertex $v\in V$ and parameter $\gamma\le\frac{\alpha-1}{8}$. We argue that the ball $B=\ball_G(v,\gamma \Delta)$ is fully contained in $P(v)$ with probability at least $e^{-4\gamma\cdot\lambda}$. We define $F_i$ as the event that some vertex of $B$ belongs to $C_i$, but $B$ does not intersect any cluster $C_j$ for $j < i$, i.e., $C_i$ is the first cluster to intersect $B$. That is, $B\cap C_{i}\ne\emptyset$ and for all $j<i$, $B\cap C_j=\emptyset$. Denote by $\mathcal{C}_i$ the event that $\mathcal{F}_i$ occurred and $B$ is cut by $C_i$ (i.e. $B\nsubseteq C_i$).

	Let $v_B\in B$ be the vertex closest to the root of $T$ (w.r.t distances in $T$). By \Cref{obs:maximal}, for every vertex $u\in B$, it holds that $u\preceq v_B$. Let $x_B\in N$ be the center that is an ancestor of $v_B$ and minimizes $\ball_{G[V_{\preceq x}]}(x,v_B)$.  Note that $v_B$ will join the cluster of $x_B$, if it did not join any other cluster. It follows that no descendant of $x_B$ can be the center of the first cluster having a non-trivial intersection with $B$. This implies that for every center $x_i\notin N_{x_B\preceq}$, $\Pr[\mathcal{F}_i]=0$.
	
	\begin{claim}\label{clm:NB-size} Let $N_B$ be the set of centers $x_i$ for which $\Pr[\mathcal{F}_i]>0$. Then $|N_B|\leq \tau$.
	\end{claim}
	\begin{subproof}
		Observe that $\mathcal{F}_i$ can have non-zero probability only if $x_i$ is an ancestor of $x_B$ and that $d_{G[V_{\preceq x_i}]}(x_{i},z) \leq R_i\leq \beta\Delta$ for some vertex $z\in B$. As  all vertices of $B$ are descendants of $x_i$,  $G[B]$ is a subgraph of $G[V_{\preceq x_i}]$. It follows from the triangle inequality that:
		\begin{align}
			d_{G[V_{\preceq x_{i}}]}(x_{i},v_{B}) & \leq d_{G[V_{\preceq x_{i}}]}(x_{i},z)+d_{G[B]}(z,v_{B})\nonumber \\
			& \leq(\beta+2\gamma)\Delta=\left(\frac{\alpha+1}{2}+2\cdot\frac{\alpha-1}{8}\right)\cdot\Delta<\alpha\Delta\label{eq:dist-xi-vB}
		\end{align}
		As $N$ is a $(\tau,\alpha,\Delta)$-tree-ordered net, there are at most $\tau$ centers in $N_i$ satisfying \cref{eq:dist-xi-vB}, implying the claim. 
	\end{subproof}	
	
	We continue by bounding the probability of a cut by each center.
	
	\begin{claim}
		For every $i$, 
		$\Pr\left[\mathcal{C}_{i}\right]\le\left(1-e^{-2\gamma\cdot\lambda}\right)\cdot\left(\Pr\left[\mathcal{F}_{i}\right]+\frac{1}{e^{(\beta-1)\cdot\lambda}-1}\right)$.
	\end{claim}
	\begin{subproof}
		We assume that by the round $i$, no vertex in $B$ is clustered, and that $d_{G[V_{\preceq x_i}]}(x_{i},B)\le\beta\Delta$, as otherwise $\Pr[\mathcal{C}_i]=\Pr[\mathcal{F}_i]=0$ and the proof follows.
		Let $\rho$ be the minimal value of $\delta_{i}$ such that if $\delta_{i}\ge\rho$, some vertex of $B$ will join $C_i$.
		Formally $\rho=\frac{1}{\Delta}\cdot d_{G[V_{\preceq x_i}]}(x_{i},B)$.
		By our assumption, $\rho\le\beta$. Set $\tilde{\rho}=\max\{\rho,1\}$. We have:
		\[
		\Pr\left[\mathcal{F}_{i}\right]=\Pr\left[\delta_{i}\ge\rho\right]=\int_{\tilde{\rho}}^{\beta}\frac{\lambda\cdot e^{-\lambda y}}{e^{-\lambda}-e^{-\beta\lambda}}dy=\frac{e^{-\tilde{\rho}\cdot\lambda}-e^{-\beta\lambda}}{e^{-\lambda}-e^{-\beta\lambda}}~.
		\]
		Let $v_i\in B$ by the closest vertex to $x_i$ w.r.t. $G[V_{\preceq x_i}]$. Note that $d_{G[V_{\preceq x_i}]}(x_{i},v_i)=\rho\cdot\Delta$.
		Then for every $u\in B$ it holds that
		\[
		d_{G[V_{\preceq x_i}]}(u,x_{i})\le d_{G[V_{\preceq x_i}]}(v_{i},x_{i})+2\gamma\Delta=\left(\rho+2\gamma\right)\cdot\Delta~.
		\]
		Therefore,  if $\delta_{i}\ge\rho+2\gamma$, the entire ball $B$ will be contained in $C_i$.
		We conclude that:
		\begin{align*}
			\Pr\left[\mathcal{C}_{i}\right] & \le\Pr\left[\rho\le\delta_{i}<\rho+2\gamma\right]\\
			& =\int_{\tilde{\rho}}^{\min\left\{ \beta,\rho+2\gamma\right\} }\frac{\lambda\cdot e^{-\lambda y}}{e^{-\lambda}-e^{-\beta\lambda}}dy\\
			& \le\frac{e^{-\tilde{\rho}\cdot\lambda}-e^{-\left(\tilde{\rho}+2\gamma\right)\cdot\lambda}}{e^{-\lambda}-e^{-\beta\lambda}}\\
			& =\left(1-e^{-2\gamma\cdot\lambda}\right)\cdot\frac{e^{-\tilde{\rho}\cdot\lambda}}{e^{-\lambda}-e^{-\beta\lambda}}\\
			& =\left(1-e^{-2\gamma\cdot\lambda}\right)\cdot\left(\Pr\left[\mathcal{F}_{i}\right]+\frac{e^{-\beta\lambda}}{e^{-\lambda}-e^{-\beta\lambda}}\right)\\
			& =\left(1-e^{-2\gamma\cdot\lambda}\right)\cdot\left(\Pr\left[\mathcal{F}_{i}\right]+\frac{1}{e^{(\beta-1)\cdot\lambda}-1}\right)~.\qedhere
		\end{align*}
	\end{subproof}
	
	We now bound the probability that the ball $B$ is cut. As the events $\{\cF_i\}_{x_i\in N}$ are disjoint, we have that
	\begin{align*}
		\Pr\left[\cup_{i}\mathcal{C}_{i}\right]=\sum_{x_{i}\in N_{B}}\Pr\left[\mathcal{C}_{i}\right] & \le\left(1-e^{-2\gamma\cdot\lambda}\right)\cdot\sum_{x_{i}\in N_{B}}\left(\Pr\left[\mathcal{F}_{i}\right]+\frac{1}{e^{(\beta-1)\cdot\lambda}-1}\right)\\
		& \le\left(1-e^{-2\gamma\cdot\lambda}\right)\cdot\left(1+\frac{\tau}{e^{(\beta-1)\cdot\lambda}-1}\right)\qquad\mbox{(by \Cref{clm:NB-size})}\\
		& \le\left(1-e^{-2\gamma\cdot\lambda}\right)\cdot\left(1+e^{-2\gamma\cdot\lambda}\right)=1-e^{-4\gamma\cdot\lambda}~,
	\end{align*}
	where the last inequality follows as
	\[
	e^{-2\gamma\lambda}=\frac{e^{-2\gamma\lambda}\left(e^{(\beta-1)\cdot\lambda}-1\right)}{e^{(\beta-1)\cdot\lambda}-1}\overset{(\#)}\ge\frac{e^{-2\gamma\lambda}\cdot e^{(\beta-1)\cdot\lambda}\cdot\frac{1}{2}}{e^{(\beta-1)\cdot\lambda}-1}\overset{(*)}{\ge}\frac{\frac{1}{2}\cdot e^{\frac{\alpha-1}{4}\cdot\lambda}}{e^{(\beta-1)\cdot\lambda}-1}=\frac{\tau}{e^{(\beta-1)\cdot\lambda}-1}~.
	\]
    The inequality $^{(\#)}$ only holds if $(\beta - 1)\lambda \geq \ln(2)$, which is the case since $(\beta - 1)\alpha = \left( \frac{\alpha + 1}{2} - 1 \right) \cdot \frac{4}{\alpha - 1} \cdot \ln(2\tau) = 2 \ln(2\tau) \geq \ln(2)$.
    
	The inequality $^{(*)}$ holds as $\beta-1-2\gamma\ge\frac{\alpha+1}{2}-1-2\cdot\frac{\alpha-1}{8}=\frac{\alpha-1}{4}$.

	To conclude, we obtain a partition where each cluster has diameter at most $2\beta\Delta\le(\alpha+1)\cdot\Delta$, and for every $\gamma\le\frac{\alpha-1}{8}$, every ball of radius
	$\gamma\cdot\Delta=\frac{\gamma}{\alpha+1}\cdot(\alpha+1)\cdot\Delta$
	is cut with probability at most $e^{-4\gamma\cdot\lambda}=e^{-\frac{\gamma}{\alpha+1}\cdot(\alpha+1)\cdot4\cdot\lambda}$.
	Thus the padding parameter is $(\alpha+1)\cdot4\cdot\lambda=\frac{\alpha+1}{\alpha-1}\cdot16\cdot\ln(2\tau)$,
	while the guarantee holds for balls of radius up to $\frac{1}{\alpha+1}\cdot\frac{\alpha-1}{8}$.
	The lemma now follows.
\end{proof}

\section{Tree-Ordered Nets for Graphs of Bounded Tree-partition Width}\label{sec:treenet-tp}

In this section, we show that graphs of bounded tree-partition width have small tree-ordered nets as claimed in \Cref{lem:NetMainLemma}, which we restate here for convenience. 

\NetMainLemma*

Herein, let $\mathcal{T}$ be a tree partition of $G$ of width $\tp$. Recall that in the definition of a tree order $\preceq$ realized by a tree $T$ and a map $\varphi: V\rightarrow V(T)$, $\preceq$  is a partial order, which means $\preceq$ is antisymmetric. This means $\varphi$ is injective, i.e., two different vertices in $V$ will be mapped to two distinct vertices in $V(T)$. In our construction presented below, it is convenient to drop the antisymmetric property to allow two distinct vertices to be mapped to the same vertex in $V(T)$. We call a tree order $\preceq$ without the antisymmetric property a \emph{semi-tree order}. 

We then can extend the notion of a tree-ordered net to a semi-tree-ordered net naturally:  a triple $(N,T,\varphi)$ is a $(\tau,\alpha,\Delta)$-semi-tree-ordered net if $T$ and $\varphi$ define a semi-tree order on $V$, and $N$ satisfies the covering  and packing properties as in \Cref{def:TreeOrderNet}. Specifically:
\begin{itemize}
	\item Covering: for every $v\in V$ there is $x\in N$ such that $\varphi(v)\preceq\varphi(x)$, and $d_{G[V_{\preceq x}]}(v,x)\le \Delta$ for the set $V_{\preceq x}=\{u\mid \varphi(u)\preceq\varphi(x)\}$.
	\item Packing: for a vertex $v\in V$, denote by $N_{v \preceq}^{\alpha\Delta}=\{x\in N\cap V_{v\preceq}\mid d_{G[V_{\preceq x}]}(v,x)\le \alpha\Delta\}$ the set of ancestor centers of $v$ at distance at most $\alpha\Delta$ (here  $V_{v\preceq }=\{u\mid \varphi(v)\preceq\varphi(u)\}$). Then $|N_{v \preceq}^{\alpha\Delta}|\le\tau$.

\end{itemize}
 In \Cref{subsec:semi-tree-ordered}, we show that graphs of tree-partition width $\tp$ admit a good semi-tree-ordered net. 

\begin{lemma} \label{lem:SeminTreeOrderNetLemma}	
	Every weighted graph  $G=(V,E,w)$ with a tree-partition width $\tp$ admits a $(\poly(\tp),3,\Delta)$-semi-tree-ordered net, for every $\Delta>0$.    
\end{lemma}

Given \Cref{lem:SeminTreeOrderNetLemma}, we now show how to construct a good tree-ordered net as claimed in  \Cref{lem:NetMainLemma}.

\begin{proof}[Proof of Lemma~\ref{lem:NetMainLemma}] Let $(N,\hat{T},\hat{\varphi})$ be a $(\poly(\tp),3,\Delta)$-semi-tree-ordered net of $G$. For each vertex $\hat{x}\in \hat{T}$, let $\hat{\varphi}^{-1}(\hat{x})\subseteq V$ be the set of vertices in $G$ that are mapped to $\hat{x}$.  That is,  $\hat{\varphi}^{-1}(x) = \{v\in V: \hat{\varphi}(v) = \hat{x}\}$. Roughly speaking, to construct a tree order $T$ for $G$,  we simply replace $\hat{x}$ in $\hat{T}$ by a path, say $P_{\hat{x}}$, composed of vertices in  $\hat{\varphi}^{-1}(\hat{x})$. The packing property remains the same, but the covering property might not hold if vertices in  $\hat{\varphi}^{-1}(\hat{x})$ are ordered arbitrarily along $P_{\hat{x}}$. Our idea to guarantee the packing property is to place net points in $\hat{\varphi}^{-1}(\hat{x})$ closer to the root of the tree.

We now formally describe the construction of $T$. Let $N_{\hat{x}} = N\cap \hat{\varphi}^{-1}(x)$. $N_{\hat{x}}$ is the set of net points that are mapped to $\hat{x}$ by $\hat{\varphi}$. Let $P_{\hat{x}}$ be a \emph{rooted} path created from vertices in $\hat{\varphi}^{-1}(x)$ where all vertices in   $\hat{\varphi}^{-1}(x)\setminus N_{\hat{x}}$ are descendants of every vertex in $N_{\hat{x}}$. Note that if $N_{\hat{x}}\not=\emptyset$, then the root of $P_{\hat{x}}$ is a vertex in $N_{\hat{x}}$. Vertices in $N_{\hat{x}}$ are ordered arbitrarily in  $P_{\hat{x}}$ and vertices in $\hat{\varphi}^{-1}(x)\setminus N_{\hat{x}}$ are also ordered arbitrarily. (If $\hat{\varphi}^{-1}(x) = \emptyset$, then $P_{\hat{x}}$ will be a single vertex that does not correspond to any vertex in $V(G)$.) Then we create the tree $T$ by connecting all the paths $\{P_{\hat{x}}: \hat{x} \in V(\hat{T})\}$ in the following way: if $(\hat{x},\hat{y})$ is an edge in $\hat{T}$ such that $x$ is the parent of $y$, then we connect the root of $P_{\hat{y}}$ to the (only) leaf of $P_{\hat{x}}$. The bijection between $P_{\hat{x}}$ and $\hat{\varphi}^{-1}(x)$, unless when $\hat{\varphi}^{-1}(x) = \emptyset$, naturally induce an injective map $\varphi:V(G)\rightarrow T$. 

Let $\preceq_T$ ($\preceq_{\hat{T}}$)  be the tree order (semi-tree order resp.) induced by  $T$ and $\varphi$ ($\hat{T}$ and $\hat{\varphi}$ resp.). To show that $T$ and $\varphi$ induce a tree order, it remains to show the validity: for every edge $\{u,v\} \in E$, either $u \preceq_T v$  or $v\preceq_T u$. Since $\preceq_{\hat{T}}$  is a semi-tree ordered, w.l.o.g., we can assume that $\hat{\varphi}(u)$ is an ancestor of $\hat{\varphi}(v)$; it is possible that $\hat{\varphi}(u) = \hat{\varphi}(v)$.  If $\hat{\varphi}(u) \not=\hat{\varphi}(v)$, then every vertex in $P_{\hat{\varphi}(u)}$ will be an ancestor of every vertex in $P_{\hat{\varphi}(v)}$, implying that $\varphi(u)$ is an ancestor of $\varphi(v)$ and hence  $u \preceq_T v$. If $\hat{\varphi}(u) =\hat{\varphi}(v)$, then $u$ and $v$ belong to the same path $P_{\hat{\varphi}(u)}$ and hence either $u \preceq_T v$ or $v \preceq_T u$. Thus, the validity follows. 

Finally, we show that $(N,T,\varphi)$ is a $(\poly(\tp),3,\Delta)$-tree-ordered net. Observe by the construction of $T$ that:

\begin{observation}\label{obs:T-order} if $u\preceq_T v$ then $u\preceq_{\hat{T}}v$.
\end{observation}

The converse of \Cref{obs:T-order} might not be true, specifically in the case when $u$ and $v$ are mapped to the same vertex in $\hat{T}$ by $\hat{\varphi}$.  By \Cref{obs:T-order}, for any vertex $x \in V$, $V_{\preceq_T x}\subseteq V_{\preceq_{\hat{T}}x}$. Thus, for any $v$ and $x$ such that $v\preceq_{T}x$, if  $d_{G[V_{\preceq_{T} x}]}(v,x) \leq 3\Delta$, then $d_{G[V_{\preceq_{\hat{T}} x}]}(v,x) \leq 3\Delta$. Therefore, $N^{3\Delta}_{v\preceq_T }\subseteq N^{3\Delta}_{v\preceq_{\hat{T}}}$ and hence $|N^{3\Delta}_{v\preceq_T }| \leq | N^{3\Delta}_{v\preceq_{\hat{T}}}| = \poly(\tp)$, giving the packing property.

Next, we show the covering property. By the covering property of $(N,\hat{T},\hat{\varphi})$, for every $v\in V$, there exists $x\in N_{v\preceq_{\hat{T}}}$ such that $d_{G[V_{\preceq_{\hat{T}} x}]}(x,v)\leq \Delta$.

Let $Q$ be a shortest path in $G[V_{\preceq_{\hat{T}} x}]$ from $v$ to $x$. Note that $Q$ might contain vertices that are not in  $G[V_{\preceq_{T} x}]$. 

\begin{claim}\label{clm:net-close} If $Q$ contains $y \not\in V_{\preceq_{T} x}$, then $y \in N$ and $x\preceq_T y$.
\end{claim}
\begin{subproof}
    Since $y \preceq_{\hat{T}}x$, as $y \not\in V_{\preceq_{T} x}$, it must be the case that $\hat{\varphi}(y) = \hat{\varphi}(x)$ by the construction of $T$. Let $\hat{x} = \hat{\varphi}(x)$. Since net points in $\hat{\varphi}^{-1}(\hat{x})$ are placed closer to the root in $P_{\hat{x}}$ and $y\not\in V_{\preceq_{T} x}$, $y$ must also be a net point and $x\preceq_T y$.
\end{subproof}

Let $z \in N\cap Q$ such that $z$ has the highest order in $T$. By  \Cref{clm:net-close}, every $y \in Q[v,z]$ satisfies $y \in  V_{\preceq_{T} z} $. This implies $z\in N_{\preceq_T v}$  and $G[V_{\preceq_{T} z}](v,z) \leq w(Q[x,z])\leq w(Q) \leq \Delta$, giving the covering property. 
\end{proof}

\subsection{The (Semi-)Tree-Ordered Net}\label{subsec:semi-tree-ordered}

 For a bag $B$ in $\mathcal{T}$, we define $\mt_B$ to be the subtree of $\mt$ rooted at $B$. For a subtree $\mt'$ of $\mathcal{T}$, we define $V[\mt']$ to be the union of all bags in $\mt'$; that is, $ V[\mt'] = \cup_{B\in \mt'}B$. 

\medskip
\hypertarget{alg:vertexOrdering}{\paragraph*{Constructing Cores.~}} We construct a set $\coreSet$ of subsets of the vertex set $V$ whose union covers $V$; see \Cref{alg:CoreFinding}.  Each set in $\coreSet$ is called a \emph{core}. We then construct a tree ordering of the vertices of $G$ using these sets. The cores are built in order to define a hierarchical structure on the graph that underlies the tree order $\preceq$. This structure lets us assign each vertex to a center bag and construct the map $\varphi$ used in the (semi-)tree-ordered net. The core construction is key to proving both the covering and packing properties required for our decomposition results.

We describe the notation used in \Cref{alg:CoreFinding}.  For a subset $U\subseteq V$ of vertices, we define $\ball_G(U,\Delta) = \cup_{u\in U}\ball_G(u,\Delta)$ to be the ball of radius $\Delta$ centered at $U$. We say that a vertex is \emph{covered} if it is part of at least one core constructed so far. For any subset  $X\subseteq V$ of vertices, we use $\uncovset(X)$ to denote the set of uncovered vertices of $X$.  We say that a bag $\mathcal{T}$ is \emph{covered} if all the vertices in the bag are covered;  a bag is \emph{uncovered} if at least one vertex in the bag is uncovered.

The algorithm proceeds in \emph{rounds}---the while loop in \cref{line:uncover}--- covers vertices in each round and continues until all vertices are covered. During each round, we process each connected component $\mt'$ of the forest induced on $\mathcal{T}$ by the uncovered bags of $\mathcal{T}$. Note that $\mt$ is rooted and  $\mt'$ is a rooted subtree of $\mathcal{T}$. The algorithm then works on $\mt'$ in a top-down manner. The basic idea is to take an unvisited bag $B$ of $\mathcal{T}'$ (initially all bags in $\mt'$ are marked unvisited) closest to the root (\cref{line:root-most}), carve a ball of radius $\Delta$ centered at the \emph{uncovered vertices} of $B$ in a subgraph $H$ (defined in \cref{line:attachment-usage}) of $G$ as a new core $R$ (\cref{line:add-core}), mark all bags intersecting with $R$ as visited, and repeat.  We call the bag $B$ in \cref{line:root-most} the \emph{center bag} of $R$, and the uncovered vertices in $B$ are called the \emph{centers} of the core $R$. Graph $H$ is the subgraph of $G$ induced by uncovered vertices in the subtree of $\mt'$ rooted at $B$ and \emph{the attachments} of the bags of this subtree---the set $A(X)$ associated with each bag $X$. We will clarify the role of the attachments below. See \Cref{fig:phase1} for an illustration. 

For each bag $B$ of the tree partition $T$, we associate a subset of vertices of $V[T_B]$, defined as the union of all bags in the subtree $T_B$ rooted at $B$. This set is called the \emph{attachment} of $B$, and is denoted by $A(B)$. Attachments allow cores created in a given round to include vertices from previously formed cores. Intuitively, the attachment $A(B)$ collects the portions of cores (from earlier rounds) whose center bags are descendants of $B$. However, attachments may also be updated dynamically during the algorithm, as we describe next.

For a subtree $\mt'$ of $\mathcal{T}$, we use $A[\mt']$ to denote the union of attachments of all bags of $\mt'$. Formally, $A[\mt'] = \cup_{B\in \mt'}A(B)$. When forming a core from a center bag $B$ in a connected component $\mt'$,  we construct a graph $H$, which is a subgraph $G$ induced by uncovered vertices in the subtree $\mt'_B$ rooted at $B$ and the attachment $A[\mt'_B]$ (\cref{line:attachment-usage}).  The core $R$ is a ball of radius at most $\Delta$ from uncovered vertices of $B$ in $H$ (\cref{line:add-core}).  We call $H$ \emph{the support graph} of $R$. Note that $R$ may contain vertices in the attachments of the bags in $\mt'_B$, and for each such bag, we remove vertices in $R$ from its attachment (\cref{line:loop-update-attach}). 

Once we construct a core $R$ from the center bag $B$, we have to update the attachment of the parent bag of $B$ to contain $R$ (\Cref{line:update-att}), unless $B$ is the root of $\mt'$. If $B$ is the root of $\mt'$, either $B$ has no parent bag---$B$ is also the root of $\mt$ in this case---or the parent bag of $B$ is covered, and hence will not belong to any connected component of uncovered bags. In both cases, the parent of $B$ will not be directly involved in any subsequent rounds; they could be involved indirectly via the attachments. It is useful to keep in mind that if $X$ is already a covered bag, the attachment update in \Cref{line:update-att} by adding a core centered at child bag $B$ has \emph{no effect} on subsequent rounds as we only consider uncovered bags.

	\begin{algorithm}[!ht]
	\caption{ $\textsc{ConstructCores}(\mt,G)$}\label{alg:CoreFinding}
	\DontPrintSemicolon
	$\coreSet \leftarrow \emptyset$\;
	$A(B)\leftarrow \emptyset$ for every bag $B\in \mt$
	\tcp*{the attachment.}

	mark every vertex of $V$ \emph{uncovered}\;
	$\mathtt{round} \leftarrow 1$\;
	\While{there is an uncovered vertex}{\label{line:uncover}
		\For{each tree $\mt'$ of the forest induced by uncovered bags of $\mt$\label{line:component}}{
			mark all bags of $\mt'$ \emph{unvisited}\;
			\While{there is an unvisited bag in $\mt'$}{
				pick the unvisited bag $B$ of $\mt'$ closest to the root \label{line:root-most}\;
				$H\leftarrow$ the graph induced in $G$ by $\uncovset(V[\mt'_B]) \cup A[\mt'_B]$ \label{line:attachment-usage}\;
				add $\ball_H(\uncovset(B),\Delta)$ as a new \emph{core} $\core$ into $\coreSet$\label{line:add-core}\;
				mark uncovered vertices of $R$ as \emph{covered}\label{line:mark-covered}\;
				mark all the bags of $\mt'$ that intersect $\core$ as \emph{visited}\label{line:mark-visited}\;
				\For{every bag $X\in \mt'_B$ s.t. $A(X)\cap R\not= \emptyset$}{ \label{line:loop-update-attach}
					$A(X)\leftarrow A(X)\setminus R$\label{line:update-AX}
					\tcp*{remove vertices of $R$ from the attachment}

				}
				\If{$B$ is not the root of $\mt'$}{
					$X\leftarrow$ parent bag of $B$ in $\mt'$\;
					$A(X)\leftarrow A(X)\cup \core$\label{line:update-att}
					\tcp*{update the attachment of $B$'s parent.}

				}
			}
		}
		$\mathtt{round} \leftarrow \mathtt{round} + 1$\;
	}
	\Return $\coreSet$\;
\end{algorithm}

\begin{figure}[ht]
	\centering
	\includegraphics[width=1\linewidth]{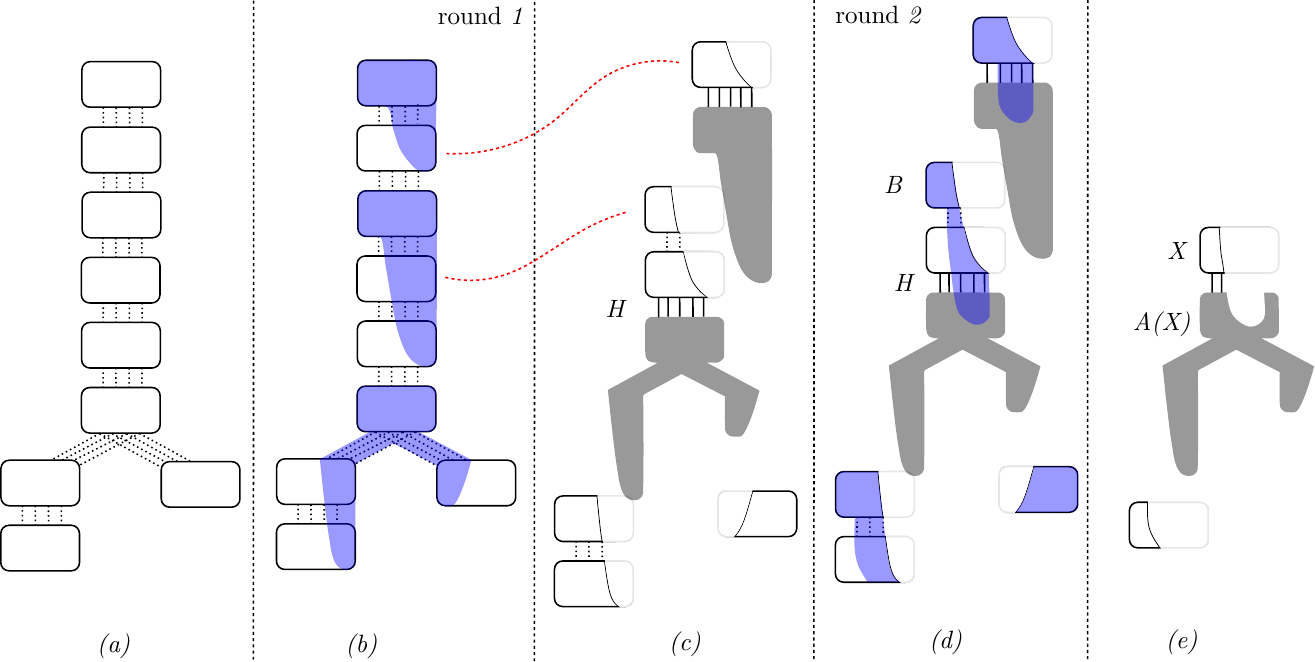}
	\caption{Illustrating two rounds of \Cref{alg:CoreFinding}. (b) In the first round, the algorithm creates  (blue) cores from $\mathcal{T}$; the attachment of every bag is $\emptyset$. (c) After the first round, uncovered bags of $\mathcal{T}$ form a forest; these are the white bags. Some bags now have non-empty attachments, illustrated by the gray-shaded regions. (d) In the second round, the algorithm considers each connected component $\mathcal{T}'$ of uncovered bags and creates cores by carving balls $\ball_H(\uncovset(B),\Delta)$ in graph $H$ from uncovered vertices of $B$. The attachment allows a core of the 2nd round to grow within a core of the 1st round. (e) The remaining uncovered bags and their attachments after round 2. A core in round 2 carves out a portion of the attachment of $X$.}
	\label{fig:phase1}
\end{figure}

As alluded above, cores in a specific round could contain vertices of cores in previous rounds via attachments. However, as we will show in the following claim that cores in the same round are vertex-disjoint.

\begin{claim}\label{clm:core-region-prop}  The cores and attachments satisfy the following properties:
	\begin{enumerate}
		\item\label{it:core-covering} $\cup_{\core \in \coreSet}\core =  V$. 
		\item\label{it:attac-des} For any bag $X$ in a connected component of unvisited bags $\mathcal{T}'$ of $\mathcal{T}$ considered in  \cref{line:component}, $A(X)$ contains vertices in the descendant bags of $X$ that are currently not in $\mathcal{T}'$. As a corollary, at any point of the algorithm, $A(X)$ only contains vertices in descendant bags of $X$, excluding~$X$.
		\item  \label{it:topdown-carving} For every core $\core\in \coreSet$ centered at a bag $B$, vertices in $R$ are in descendant bags of $B$. ($B$ is considered its own descendant.) 
		\item \label{it:covered-attac}  For every core $\core \in \coreSet$ centered at a bag $B$, if a vertex $u\in \core$ was covered before $\core$ is created, then $u$ is in the attachments of descendant bags of $B$. 
		\item \label{it:core-dijst}  Let $\core_1$ and $\core_2$ be two different cores created from the same connected component $\mt'$ in the same round. Then $V(\core_1)\cap V(\core_2) = \emptyset$. 
	\end{enumerate}    
\end{claim}
\begin{proof}
	
	We observe that the algorithm only terminates when every vertex is covered. Furthermore, the algorithm only marks an uncovered vertex as covered when it is contained in a new core (\cref{line:mark-covered}), implying \cref{it:core-covering}. 
	
	\Cref{it:attac-des} follows from the fact that whenever we update the attachment of an uncovered bag $X$ in \cref{line:update-att} by a core $R$, the center bag $B$ of $R$ become covered and hence will be disconnected from the connected component containing $X$ in the next round. 
	
	For \cref{it:topdown-carving}, we observe that the support graph $H$ of $R$ in \cref{line:attachment-usage} contains uncovered vertices in descendant bags of $B$ and their attachments. By \cref{it:attac-des}, the attachment of a bag $X$ only contains vertices in the descendant bags of $X$. Thus, vertices in $H$ are in descendant bags of $B$, as claimed.
	
	For \cref{it:covered-attac}, observe that the only way for a covered vertex to be considered in subsequent rounds is via the attachments. Thus, \cref{it:covered-attac} follows from \cref{it:topdown-carving}. 
	
	For the last \cref{it:core-dijst}, suppose otherwise: there is a vertex $v$ in a bag $Z$ such that $v\in \core_1\cap\core_2$. W.l.o.g, we assume that $\core_1$ is created before $\core_2$. Let $B_1$ ($B_2$, resp.) be the center bag of $\core_1$ ($\core_2$, resp.). Then $B_1$ is the ancestor of $B_2$ by construction, and furthermore, by \cref{it:topdown-carving},  $Z$ is descendant of both $B_1$ and $B_2$. 
	
	If $Z \in \mt'$, then when $\core_1$ was created, all the bags on the path from $B_1$ to $Z$ will be marked visited, and hence $B_2$ will also be marked visited, contradicting that $\core_2$ was created from $B_2$. Otherwise, $Z \not \in \mt'$ and that means $v$ belongs to the attachment $A(Y)$ of some bag $Y\in \mt'$. Since all bags on the path from $B_1$ to $Y$ are marked when $\core_1$ was created, $B_2$ is not an ancestor of $Y$. But this means $v$ cannot belong to a descendant bag of $B_2$, contradicting \cref{it:topdown-carving}. 
\end{proof}

A crucial property of the core construction algorithm used in our analysis is that it has a natural hierarchy of clusters associated with it. To define this hierarchy, we need some new notation.  Let $\mt'$ be a connected component of uncovered bags in a specific round. Let $C_{\mt'} = \uncovset(V[\mt']) \cup A[\mt']$ be a set of vertices, called the \emph{cluster} associated with $\mt'$. Note that $V(H)\subseteq C_{\mt'}$. The following lemma implies that clusters over different rounds form a hierarchy.

\begin{lemma}\label{lm:hierarchy} Let $\mt_1$ and $\mt_2$ be two different connected components of uncovered bags. If $\mt_1$ and $\mt_2$  belong to the same round, then  $C_{\mt_1}\cap C_{\mt_2} = \emptyset$. If $\mt_1$ and $\mt_2$ belong to two consecutive rounds such that $\mt_2$ is a subtree of $\mt_1$, then $C_{\mt_2}\subseteq C_{\mt_1}$. 
\end{lemma}
\begin{proof}
	We prove the lemma by induction. The base case holds since in round 1, we only have a single tree $T$ whose associated cluster $C_{\mathcal{T}} = V(G)$.  Let $\mt'$ be a connected component of uncovered bags at round $i$, and $\mt_1,\mt_2$ be two different components of uncovered bags that are subtrees of $\mt'$ in round $i+1$.  It suffices to show that $C_{\mt_i}\subseteq C_{\mt'}$ for any $i = 1,2$, and $C_{\mt_1}\cap C_{\mt_2} = \emptyset$.
	
	Observe that in a round $i$, the attachment of a bag $X$ either shrinks (due to the removal in \cref{line:update-AX}) or grows by the addition of the cores. As any core  $\core_1$ created from $\mt'$ is a subset of the cluster $C_{\mt'}$, $A(X)$ remains a subset of  $C_{\mt'}$ in both cases. This means all attachments of bags in $\mt_1$ are subsets of $C_{\mt'}$, implying that $C_{\mt_1}\subseteq C_{\mt'}$. The same argument gives  $C_{\mt_2}\subseteq C_{\mt'}$.
	
	Let $B_1$ and $B_2$ be the root bag of $\mt_1$ and $\mt_2$, respectively. Assume w.l.o.g that $B_1$ is an ancestor of $B_2$. Let $X$ be the (only) leaf of $B_1$ that is the ancestor of $B_2$.  By \cref{it:covered-attac} of \Cref{clm:core-region-prop}, only the attachment of $X$ could possibly have a non-empty intersection with $C_{\mt_2}$. Suppose that there exists a vertex $v\in A(X)\cap C_{\mt_2}$. Let $B$ be the child bag of $X$ that is an ancestor of $B_2$. Then we have that $B \in T'$, and $B$ is the center bag of a core $R$ created in round $i$. Observe that $v\in \core$. When $\core$ was created, all the uncovered vertices in $\mt' \cap \core$ were marked covered. This means $\core\cap \uncovset(V[\mt_2]) = \emptyset$. Thus,  $v\in \core \cap A(Y)$ for some bag $Y\in \mt'$ in round $i+1$. However, by \cref{line:update-AX}, $v$ will be removed from $A(Y)$ in round $i$ and hence will not be present in $A(Y)$ in round $i+1$, a contradiction. Thus,   $C_{\mt_1}\cap C_{\mt_2} = \emptyset$.
\end{proof}

\paragraph*{The Semi-tree-ordered Net.~} We now construct a tree ordering of the vertices using $\coreSet$.   We define the rank of a core $\core$, denoted by $\rank(\core)$, to be the round number when $\core$ is constructed. That is, $\rank(\core) =i$ if $\core$ was constructed in the $i$-th round.  Lastly, for a vertex $v$, we define $Q(v)$ to be the center bag of the \emph{smallest-rank core} containing $v$. We observe that:
\begin{observation}\label{obs:smallest-rank}
	The core $R$ covering $v$ for the first time is the smallest-rank core containing~$v$.
\end{observation}
We now define the tree ordering of vertices, and the tree-ordered net as follows. 

\begin{definition}[Tree Ordering and Tree-Ordered Net]\label{def:vertexOrdering}
Let $T$ be the rooted tree that is \emph{isomorphic} to the tree partition $\mathcal{T}$ of $G$; each node $x\in T$ corresponds to a bag $B_x$ in $\mathcal{T}$. 
        \begin{itemize}
		\item  The map $\varphi: V\rightarrow V(T)$ maps each vertex $v$ to a node $\varphi(v)$ such that its corresponding bag $B_{\varphi(v)}$ is exactly $Q(v)$. This map naturally induces a semi-tree ordering of vertices in $V$:   $u \preceq v$ if and only if $Q(v)$ is an ancestor of $Q(u)$ in the tree partition $\mathcal{T}$.
		\item  The tree-ordered net $N$ is the union of all the centers of the cores in $\coreSet$. Recall that the centers of a core is the set of uncovered vertices of the center bag in \cref{line:add-core}. 
	\end{itemize}
	The semi-tree-ordered net for $G$ is $(N, T, \varphi)$. 
\end{definition}

 The following lemma stems directly from the description of the algorithm and the discussion above.

\begin{lemma}
	\label{lem:vertexIsInCore}
	Given the tree partition $\mt$ of $G$, the semi-tree-ordered net $(N,T,\varphi)$ can be constructed in polynomial time.
\end{lemma}

We can see that the tree ordered induced by $T$ and  $\varphi$ is a semi-tree-order since it does not have the antisymmetric property: there could be two different vertices $u$ and $v$ such that $Q(u) = Q(v)$ and hence they will be mapped to the same vertex of $T$ via $\varphi$.  

\subsection{The Analysis}

We now show all properties of the tree-ordered net $(N,T,\varphi)$ as stated in \Cref{lem:NetMainLemma}. Specifically, we will show the covering property in \Cref{lem:covering} and the packing property in \Cref{lem:packing}. We start with the covering property.

\begin{lemma}
	\label{lem:covering}
	Let $(N,T,\varphi)$ be the tree-ordered net defined in \Cref{def:vertexOrdering}.  For every vertex $v \in V$ there is an $x\in N_{v\preceq}$ such that $d_{G[V_{\preceq x}]}(v,x)\le \Delta$.
\end{lemma}
\begin{proof} By \cref{it:core-covering} of \Cref{clm:core-region-prop},  $v$ is in at least one core of $\coreSet$.  Let $R \in \coreSet$ be the core of the smallest rank containing $v$. Let $B$ be the center bag of $R$. Note that $Q(v) = B$ by definition.
	
	Let $U\subseteq \core$ be the set of uncovered vertices in $\core$ when $R$ is created; all vertices in $U$ will be marked as covered in \cref{line:mark-covered} after $R$ is created.  We observe that:
	\begin{enumerate}
		\item All vertices in $U$ are equivalent under the tree ordering $\preceq$. This is because $Q(u)=B$ for every $u\in U$  by \Cref{obs:smallest-rank} and definition of $Q(\cdot)$.
		\item  For any two vertices $z\in R\setminus U$ and $u\in U$, $z \prec u$ (i.e. $z$ a proper descendant of $u$). To see this, observe that, by  \cref{it:covered-attac} in \Cref{clm:core-region-prop},  $z$ is in the attachment of some bag $Y$, which is the descendant of the center bag $B$.   By \cref{it:attac-des} and the definition of $Q(\cdot)$, $Q(z)$ is a (proper) descendant of $Y$. Therefore, $Q(z)$ is a descendant of $B$, which is $Q(u)$, implying that $z\prec u$.
	\end{enumerate}
	Let $H$ be the support graph (in \cref{line:attachment-usage}) of $R$. By the construction of $\core$ (\cref{line:add-core}), $d_{H[R]}(v,U)\leq \Delta$.  This means there exists an $x\in U$ such that  $d_{H[R]}(v,x)\leq \Delta$. The two observations above imply that $H[R] \subseteq G[V_{\preceq x}]$. Thus, $d_{G[V_{\preceq x}]}(v,x)\leq \Delta$ as desired.       
\end{proof} 

The packing property is substantially more difficult to prove. Our argument goes roughly as follows. First, we establish several structural properties of the cores. In particular, we show that for every bag $B$ in $T$, the number of cores intersecting $B$ is at most $O(\mathrm{tp}^2)$, and that each vertex $v$ belongs to at most $O(\mathrm{tp})$ cores (\Cref{cor:core-bag-intersect}). This allows us to bound $\left|N_{v\preceq}^{2\Delta}\right|$ via bounding the number of cores that contain at least one vertex in $N_{v\preceq}^{2\Delta}$. This set of cores can be partitioned into two sets: those that contain $v$---there are only $O(\tp)$ of them---and those that do not contain $v$. To bound the size of the latter set, we basically construct a sequence of cores of strictly increasing ranks $R_1^*, R_2^*, \ldots, R_\ell^*$ for $\ell \leq \tp$, and for each $R_j^*$, show that there are only $O(\tp^2)$ ancestral cores of $R_j^*$ that are not ancestors of lower ranked cores in the sequence. This implies a bound of $O(\tp^3)$ on the set of cores, giving the packing property.

We begin by analyzing several properties of the cores. 

\begin{lemma}  \label{lemma::coresDisjointRank}The cores of the same rank are vertex-disjoint.  
\end{lemma}
\begin{proof}  
 Let $R_1$ and $R_2$ be two cores of the same rank, say $r$. Let $B_k$ be the center bags of $R_k$ for $k=1,2$.  If $B_1$ and $B_2$ belong to two different connected components of uncovered bags in round $r$, then by \Cref{lm:hierarchy}, $R_1\cap R_2 = \emptyset$. Thus, we only consider the complementary case where $B_1$ and $B_2$ belong to the same connected component.  
	
	Suppose for contradiction that there exists $v\in R_1\cap R_2$. By \cref{it:topdown-carving} in \Cref{clm:core-region-prop}, $R_k$ only contains vertices in descendant bags of $B_k$.  As $R_1\cap R_2 \not=\emptyset$,  either $B_1$ is an ancestor of $B_2$ or $B_2$ is ancestor of $B_1$. W.l.o.g., we assume that $B_1$ is an ancestor of $B_2$. Let $Y$ be the bag containing $v$. Then $Y$ is the descendant of both $B_1$ and $B_2$. As endpoints of edges of $G$ are either in the same bag or in two adjacent bags of the tree partition, $R_1\cap Y \not=\emptyset$ implies that $R_1\cap B_2\not=\emptyset$. Thus, $B_2$ will be marked as visited in \cref{line:mark-visited} of the algorithm, and hence $R_2$ will not have the same rank as $R_1$, a contradiction.  
\end{proof}

\begin{lemma}\label{lem:numberiterations}
	The algorithm has at most $\tp$ iterations. Therefore, $\rank(\core) \leq \tp$ for every $\core \in \coreSet$.
\end{lemma}
\begin{proof}
	Let $B'$ be any bag of $\mt$. We claim that in every iteration, at least one uncovered vertex of $B'$ gets \emph{covered}. Thus, after $\tp$ iterations, every vertex of $B'$ is covered, and hence the algorithm will terminate.
	
	Consider an arbitrary iteration where $B'$ remains uncovered; that is, at least some vertex of $B'$ is uncovered. Let $\mt'$ be the connected component of uncovered bags of $\mt$ containing $B'$. If $B'$ got picked in \cref{line:add-core}, then all uncovered vertices of $B'$ are marked as covered (in \cref{line:mark-covered}), and hence the claim holds. Otherwise, $B'$ is marked visited in \cref{line:mark-visited} when a core $\core \in \coreSet$ is created. Let $H$ be the support graph of $\core$. By \cref{it:attac-des} in \Cref{clm:core-region-prop}, no bag in $\mt'$ contains a vertex of $B'$. Thus, $H$ only contains uncovered vertices of $B'$. As $\core\cap B' \not=\emptyset$, $\core$ contains at least one uncovered vertex of $B'$, which will be marked as covered in \cref{line:mark-covered}; the claim holds.  
\end{proof}

We obtain the following corollary of \Cref{lemma::coresDisjointRank} and \Cref{lem:numberiterations}.

\begin{corollary}\label{cor:core-bag-intersect} Each vertex is contained in at most $\tp$ cores. Furthermore, the number of cores intersecting any bag is at most $\tp^2$.
\end{corollary}
\begin{proof}
	By \Cref{lemma::coresDisjointRank},  cores of the same rank are vertex-disjoint. Thus, each vertex belongs to at most one core of a given rank. As there are $\tp$ different ranks by \Cref{lem:numberiterations}, each vertex belongs to at most $\tp$ cores. 
	
	Let $B$ be any bag in $\mt$. As the cores of the same rank are vertex-disjoint by \Cref{lemma::coresDisjointRank}, there is at most $\tp$ cores of a given rank intersecting $B$. As there are at most $\tp$ different ranks by  \Cref{lem:numberiterations}, the total number of cores intersecting any bag is at most $\tp^2$.
\end{proof}

We define the \emph{rank} of a vertex $v$, denoted by $\rank(v)$, to be the lowest rank among all the cores containing it: $\rank(v) = \min_{\core\in \coreSet \wedge  v\in \core}\rank(\core)$.  The following lemma also follows from the algorithm.
\begin{lemma}
	\label{lem:noncentervertices}Let $B$ be the center bag of a core $\core$, and the vertices of $U$ be its center. If $\core$ has rank $i$, then every vertex $v\in B\setminus U$ has rank strictly smaller than $i$. 
\end{lemma}
\begin{proof}
	When $\core$ is constructed (in \cref{line:add-core}), all uncovered vertices of $B$ are in $U$ and will be marked as covered afterward. Thus, before $\core$ is constructed, vertices in $B\setminus U$ must be marked covered, and furthermore, $B$ is not marked as visited. This means vertices in $B\setminus U$ are marked in previous rounds, and hence have ranks strictly smaller than $i$. 
\end{proof}

 Next, we introduce central concepts in the proof of the packing property. We say that a core $\core_1$ is an ancestor (descendant resp.) of $\core_2$ if the center-bag of $\core_1$ is an ancestor (descendant resp.) of the center bag of $\core_2$. For any bag of $\mathcal{T}$, we define its \emph{level} to be the hop distance from the root in $\mathcal{T}$. Also, we define the \emph{graph rooted at $B$} to be the subgraph of $G$ induced by the subtree of $\mathcal{T}$ rooted at $B$. For any core $\core$, we define a graph $G_{\mt}[\core]$ to be the subgraph of $G$ induced by vertices in the bags of the subtrees of $\mt$ rooted at the center-bag of $\core$; see \Cref{fig:shadow}(a). 

Let $\mathcal{A}^{\prec}(\core)$ be the set of cores that are ancestors of $\core$ and have rank strictly less than $\core$. We define the \emph{shadow domain} of a core $\core$ to be the graph obtained from $G_{\mt}[\core]$ by removing the vertices that are contained in at least one core in $\mathcal{A}^{\prec}(\core)$. The \emph{shadow} of a core $\core$, denoted by  $\shadow({\core})$, is defined as the ball of radius $\Delta$ centered around $\core$ in the shadow domain of $\core$. The \emph{strict shadow} of $\core$ is defined as its shadow minus itself, i.e., $\shadow({\core}) \setminus \core$. To bound the size of $N_{v\preceq}^{2\Delta}$, for a vertex $v$ we are interested in the cores and their shadows where $v$ is located. See \Cref{fig:shadow}(b).

\begin{figure}[ht]
	\begin{center}
		\begin{tikzpicture}
                  \node at (0,0){\includegraphics[width=0.8\linewidth]{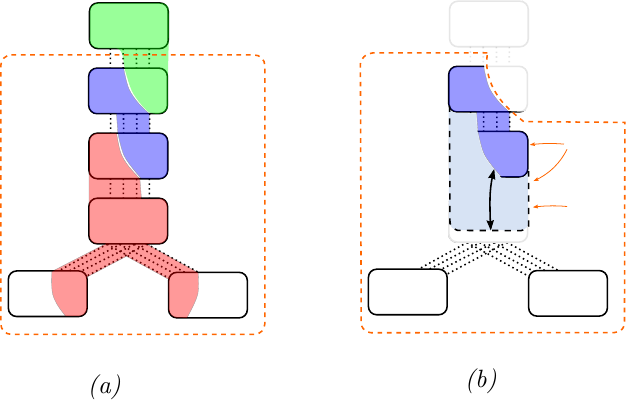}};
                 \footnotesize
			\node at (-6.2,4.55){$R_1$};
			\node at (-6.2,2.75){$R_2$};
			\node at (-6.2,1.2){$R_3$};
			\node at (-7.2,4.0){$G_{\mt}[R_2]$};
			\node at (2.0,4.6){shadow domain};
			\node at (2.0, 4.2){of $R_2$};
			\node at (3.1,2.75){$R_2$};
			\node at (7.7,1.35){shadow of $R_2$};
			\node at (8.4,-0.1){strict shadow of $R_2$};
			\node at (4.2,-0.1){$\Delta$};
		\end{tikzpicture}   
	\end{center}
	\caption{(a) Core $\core_2$ and graph $G_{\mt}[R_2]$ rooted at its center bag, (b) the shadow domain, shadow and strict shadow of core $\core_2$.}
	\label{fig:shadow}
\end{figure}

\begin{cfigure}[!ht]
{\fbox{
	\begin{minipage}{.45\textwidth}
		\footnotesize 
		$\boldsymbol{\circ}$~~ \textbf{Rank} of $\core$: the round number when $\core$ is constructed.
		
		\vspace{7pt}$\boldsymbol{\circ}$~~ \textbf{Central bag} of $\core$ - the bag from \cref{line:root-most}, from which we grew the core.
		
		\vspace{7pt}$\boldsymbol{\circ}$~~ Core $\core_{1}$ is an \textbf{ancestor} of $\core_{2}$ if the center-bag of $\core_{1}$ is an ancestor of $\core_{2}$ (independent of round). 
		
		\vspace{7pt}$\boldsymbol{\circ}$~~ The \textbf{level} of a bag $B$ is the hop distance from the root in $\mathcal{T}$. 
		
		\vspace{7pt}$\boldsymbol{\circ}$~~ The graph \textbf{rooted} at $B$ is the subgraph of $G$ induced by the subtree of $\mathcal{T}$ rooted at $B$.

	\end{minipage}
	\begin{minipage}{.05\textwidth}~
	\end{minipage}
	\begin{minipage}{.45\textwidth}
		\footnotesize

		$\boldsymbol{\circ}$~~ $G_{\mt}[\core]$ denotes to be the graph rooted at the center-bag
		of $\core$ (independent of rank).

		\vspace{5pt}$\boldsymbol{\circ}$~~ $\mathcal{A}^{\prec}(\core)$ - cores that are ancestors of $\core$
		and have rank strictly less than $\core$.
		
		\vspace{5pt}$\boldsymbol{\circ}$~~ \textbf{Shadow domain} of $\core$ is the graph obtained from $G_{\mt}[\core]$ by removing the vertices that are contained in at least one core in $\mathcal{A}^{\prec}(\core)$.
		
		\vspace{5pt}$\boldsymbol{\circ}$~~ $\shadow({\core})$: \textbf{Shadow} of $\core$ is the ball of radius $\Delta$ centered around $\core$ in the shadow domain of $\core$.
		
		\vspace{5pt}$\boldsymbol{\circ}$~~ \textbf{Strict shadow} of $\core$: the shadow minus $\core$: $\shadow({\core})\setminus\core$.
	\end{minipage}
}}
	\caption{\label{fig:PathDistortion}\small Key definitions used during the proof of \Cref{lem:vertexIsInCore}.}
\end{cfigure}

We observe the following properties of ranks and shadow domains.

\begin{lemma} \label{lemma:cores}
	Let $\core_1$ and $\core_2$ be two cores such that $ \rank(\core_1)\geq \rank(\core_2)$ and $\core_1$ is an ancestor of $\core_2$. For each $i\in \{1,2\}$, let $D_i$, $B_i$, and $Y_i$ be the shadow domain, the center bag, and the center of $R_i$, respectively. (Note that $Y_i\subseteq B_i$.) 
	\begin{enumerate}
		\item \label{item:corewithhigherrank} If $V(D_1) \cap (B_2\setminus \centre_2)\not= \emptyset$, then there exists a  core $\core_3$ that is an ancestor of $\core_2$ and descendant of $\core_1$  such that $\rank(\core_3) < \rank(\core_2)$.
		\item \label{item:shadowGraphInclusion} If  $ \rank(\core_1) > \rank(\core_2)$ and there are no cores of rank smaller than $\core_2$ whose center bag is in the path between $B_1$ and $B_2$ (in $\mt$), then $V(G_{\mt}[\core_2])\cap V(D_1) \subseteq V(D_2)$. In other words,  every vertex in the shadow domain of $\core_1$ in $G_{\mt}[\core_2]$ is in the shadow domain of $\core_2$.
	\end{enumerate}
\end{lemma}
\begin{proof} We first show \cref{item:corewithhigherrank}. Let $x$ be a vertex in $V(D_1)\cap (B_2\setminus \centre_2)$. By \Cref{lem:noncentervertices}, $\rank(x) < \rank(\core_2)$, implying that $x$ is contained in a core $\core_3$ such that $\rank(\core_3) < \rank(\core_2)$.  Furthermore, $\core_3$ is a descendant of $\core_1$ as otherwise, by the definition of shadow domain, $\core_3\cap V(D_1) = \emptyset$ and hence $x\not\in V(D_1)$, a contradiction. Also, $\core_3$ is an ancestor of $\core_2$ as the center-bag of $\core_2$ contains a vertex of $\core_3$, which is $x$. 
	
	We show \cref{item:shadowGraphInclusion} by contrapositive. Let $x$ be a vertex in $G_{\mt}[\core_2]$ such that  $x \not\in V(D_2)$. We show that $x \not\in V(D_1)$. Since $x$ is not in $V(D_2)$,  there exists an ancestor core  $\core_3$ of $\core_2$ such that $x\in \core_3$ and $\rank(\core_3) < \rank(\core_2)$. By the assumption in \cref{item:shadowGraphInclusion}, $\core_3$ is also an ancestor of $\core_1$ and hence $R_3\in \mathcal{A}^{\prec}(R_1)$. Thus,  $x \not \in V(D_1)$ as claimed. 
\end{proof}

Equipped with the lemmas above, we are finally ready to prove the packing property with $\alpha=2$.

\begin{lemma}
	\label{lem:packing} Let $(N,T,\varphi)$ be the tree-ordered net defined in \Cref{def:vertexOrdering}.
	For every vertex $v \in V$,  $\left|N_{v\preceq}^{2\Delta}\right|\le \tp^4 + \tp^2$.
\end{lemma}
\begin{proof}
	Consider a vertex $v \in V$. Recall that the net $N$ is the set of centers of all cores in $\coreSet$.  Let $\C_{\texttt{all}}$ be the set of cores whose center that have a non-empty intersection with the vertices in $N_{v\preceq}^{2\Delta}$. Instead of bounding $|N_{v\preceq}^{2\Delta}|$ directly, we bound $|\C_{\texttt{all}}|$. Note that every vertex in $N_{v\preceq}^{2\Delta}$ is in the center of one of the cores in $\C_{\texttt{all}}$. 
	
	Let $\C^{\Delta} \subseteq \C_{\texttt{all}}$ be the cores that contain $v$ and let $\C := \C_{\texttt{all}} \setminus \C^{\Delta}$.
	By \Cref{cor:core-bag-intersect}, $|\C^{\Delta}| \leq \tp$. The bulk of our proof below is to show that $|\C|\le \tp^3$. Since each core has at most $\tp$ centers, we have $\left|N_{v\preceq}^{2\Delta}\right|\le \tp\cdot |\C_{\texttt{all}}| \leq \tp^4 + \tp^2$, as claimed.

	We now focus on proving that $|\C| \le \tp^3$. Let $B_v$ be the bag containing $v$.  Let $\mP$ be the path in $\mathcal{T}$ from the root bag to $B_v$. We claim that every core in $\C$ has a center bag on $\mP$. This is because every core $\core$ in $\C$ contains a vertex, say $x$, in $N_{v\preceq}^{2\Delta}$, whose bag is an ancestor of  $B_v$. By construction in \cref{line:add-core}, every vertex in $\core$ is in a descendant of the center bag of $\core$. This means $B_v$ is a descendant of the center bag of $\core$, implying the claim.

	The rest of our proof goes as follows:
	\begin{itemize}
		\item Let $r_1$ be the lowest rank among all cores in $\C$. We will show in \Cref{lem:oner_j} below that there is only one core in $\C$ having rank $r_1$.  Let this unique core in $\C$ with rank $r_1$ be $\core^*_1$. Let $\C_1$ be the set of all cores in $\C$ that are ancestors of $\core^*_1$, including $\core^*_1$. Next, we show in \Cref{lem:intersectQ_j} that each core in $\C_1$ intersects the center of $\core^*_1$. Since the center is contained in a bag, \Cref{cor:core-bag-intersect} implies that $|\C_1|\le \tp^2$.
		\item  Next, let $\bar{\C}_1=\C\setminus \C_1$. If $\barC_1$ is non-empty, we define $r_2$ be the lowest rank among cores in $\bar{\C_1}$. Note that $r_2>r_1$. \Cref{lem:oner_j} below again implies that there is only one core in $\barC_1$ having rank $r_2$. Let this unique core in $\barC_1$ with rank $r_2$ be $\core^*_2$. Let $\C_2$ be the set of all ancestor cores $\core^*_2$ in $\barC_1$ including $\core^*_2$.  Then \Cref{lem:intersectQ_j} implies that each core in $\C_2$ intersects the center of $\core^*_2$. Since the center is contained in a bag, \Cref{cor:core-bag-intersect} gives that $|\C_2|\le \tp^2$.
		\item Inductively, we define the sequence of sets of cores $ \C_1,\C_2,\C_3,\dots ,\C_{\ell}$  and $r_1<r_2<\dots <r_{\ell}$ until $\barC_{\ell}=\barC_{\ell-1}\setminus \C_{\ell}$ is empty.  Here, $r_j$ is defined as the lowest rank among cores in $\barC_{j-1}$, and $\barC_{0} = \C$. For each $j\in [\ell]$, \Cref{lem:oner_j} implies that there exist a unique core $\core^*_j$ in $\barC_{j-1}$ with rank $r_j$. Let $\C_j$ be the set of all cores in $\barC_j$ that are ancestors of $\core^*_j$, including $\core^*_j$.  \Cref{lem:intersectQ_j} then implies that all cores in $\C_j$ intersect the center of $\core^*_j$ in. Thus, $|\C_j|\le \tp^2$ for each $j\in [\ell]$ by \Cref{cor:core-bag-intersect}. 
		By \Cref{lem:numberiterations}, $r_1<r_2<\dots <r_{\ell}\le \tp$, implying that $|\C|\le \tp^3$ as desired. 
	\end{itemize}

	For the rest of the proof, we prove two claims.

	\begin{figure}[ht]
		\centering
		\includegraphics[width=0.8\linewidth]{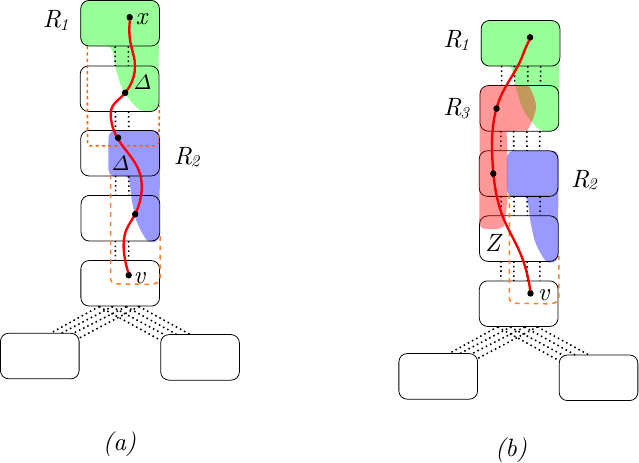}
		\caption{Illustration for the proof of \Cref{lem:oner_j}.}
		\label{fig:cores-rank}
	\end{figure}
	
	\begin{claim}\label{lem:oner_j}
		There is only one core in $\barC_{j-1}$ having rank $r_j$ for each $j\in [\ell]$.
	\end{claim}
	\begin{subproof}
		Suppose otherwise. Then there are two cores $\core_1$ and $\core_2$ in $\barC_{j-1}$ having rank $r_j$. Assume w.l.o.g. that $\core_1$ is an ancestor of $\core_2$.  By \Cref{lemma::coresDisjointRank},  $\core_1$ and $\core_2$ are disjoint as they have the same rank. Also, since $\core_1,\core_2\in \C$, the center-bags of $\core_1$ and $\core_2$ lie on the path $\mP$.
		
		Since  $\core_1$ contains a point of  $N_{v\preceq}^{2\Delta}$, $v$ is in the shadow of $\core_1$. Thus, there is a path $Z$ of length at most $2 \Delta$  that goes from the center of $\core_1$ to $v$ in the shadow domain of $\core_1$. This path $Z$ has to intersect the center bag of $\core_2$ to get to $v$.

		Suppose this intersection occurs at a vertex in the center of $\core_2$; see \Cref{fig:cores-rank}(a). Then the path $Z$ goes from the center of $\core_1$ to outside $\core_1$ and then into the center of $\core_2$ and then to outside $\core_2$. It has to go outside of $\core_2$ as $v$ is in the strict shadow of $\core_2$.
		Also, it has to go outside of $\core_1$ before entering $\core_2$ as $\core_1$ and $\core_2$ are vertex disjoint (however, it is possible that there is an edge from $\core_1$ to $\core_2$).
		This means $Z$ has a length of more than $\Delta+\Delta=2\Delta$, a contradiction.
		
		Now, suppose this intersection occurs at a vertex in the center-bag of $\core_2$ that is not in the center of $\core_2$; see \Cref{fig:cores-rank}(b).
		Then by \cref{item:corewithhigherrank} of \Cref{lemma:cores} it follows that there is at least one core that is ancestor of $\core_2$ and descendant of $\core_1$, and having rank lower than $r_j$. Let $\core_3$ be the one among such cores whose center-bag has the smallest level. 
		
		We claim that the path $Z$ intersects the center of $\core_3$. Suppose otherwise. However, the path has to intersect the center-bag of $\core_3$ to get to $v$.
		Then, by \cref{item:corewithhigherrank} (of \Cref{lemma:cores}) it follows that there is a core that is ancestor of $\core_3$ and descendant of $\core_1$, and having rank lower than $r_j$, a contradiction to the selection of $\core_3$.
		
		By definition of shadow domain, $\core_3$ is disjoint from the shadow domain of $\core_2$. Hence $v$ is not in $\core_3$. However, $v$ is contained in the shadow of $\core_3$ as the part of the path $Z$ from the center of $\core_3$ to $v$ is contained in the shadow domain of $\core_3$ and has length at most $2\Delta$. Thus, $v$ is contained in the strict shadow of $\core_3$. Since $\core_3$ is a descendant of $\core_1$, we have that $\core_3\notin \C_1\cup\C_2\cup\dots\cup\C_{j-1}$ and hence $\core_3\in \barC_{j-1}$. Thus, there is a core in $\barC_{j-1}$ that has rank lower than $r_j$, a contradiction to the choice of $r_j$.
	\end{subproof}
	
	\begin{claim}
		Each core in $\C_j$ intersects the center of $\core^*_j$ for each $j\in [\ell]$. 
		\label{lem:intersectQ_j}
	\end{claim}
	\begin{subproof}
		Suppose this is not true, and let $j$ be the minimum index for which the claim does not hold. Then there exists a core $\core \in \C_j$ that is disjoint from the center of $\core^*_j$. 
		Note that $\core^*_j$ is a descendant of $\core$ by definition of $\C_j$ and the rank of $\core^*_j$ is strictly less than the rank of $\core$ by the definition of $\core^*_j$.
		Since $v$ is in the shadow of $\core$ there is a path $Z$ of length at most $2 \Delta$ that goes from the center of $\core$ to $v$ in the shadow domain of $\core$.
		This path $Z$ has to intersect the center-bag of $\core^*_j$ to get to $v$.
		
		Suppose this intersection occurs at a vertex in the center of $\core^*_j$. Then the path $Z$ goes from the center of $\core$ outside $\core$ and then into the center of $\core^*_j$ and then to outside of $\core^*_j$. It has to go outside of $\core^*_j$ as $v$ is in the strict shadow of $\core^*_j$. Also, it has to go outside of $\core$ before entering $\core^*_j$ as $\core$ and the center of $\core^*_j$ are vertex disjoint (this is what we assumed for the sake of contradiction). 
		This means $Z$ has a length of more than $\Delta+\Delta=2\Delta$, a contradiction.
		
		Now, suppose this intersection occurs at a vertex in the center-bag of $\core^*_j$ that is not in the center of $\core^*_j$.
		Then, by \cref{item:corewithhigherrank} (of \Cref{lemma:cores}), there exist at least one core that is an ancestor of $\core^*_j$ and a descendant of $\core$ and having rank lower than $r_j$.
		Let $\core'$ be the one among such cores whose center-bag has the smallest level.
		The path $Z$ has to intersect the center of $\core'$ as otherwise
		there is a core that is an ancestor of $\core'$ and a descendant of $\core$ and having rank lower than $r_j$, contradicting the selection of $\core'$.
		
		Note that $v$ is not in $\core'$ as $\core'$ is disjoint from the shadow domain of $\core^*_j$ by definition of shadow domain.
		However, $v$ is contained in the shadow of $\core'$ as the part of the path $Z$ from center of $\core'$ to $v$ is contained in the shadow domain of $\core'$ (by \cref{item:shadowGraphInclusion} of \Cref{lemma:cores}) and has length at most $2\Delta$.
		Thus, $v$ is contained in the strict shadow of $\core'$.
		Since $j$ is the minimum index for which the statement of the claim does not hold, if $\core'\in \C_1 \cup \ldots \C_{j-1}$, then $\core \not \in \C_j$. This implies that $\core'\in \C_j$.
		Thus we have a core in $\C_j$ having rank strictly smaller than $r_j$, a contradiction.
	\end{subproof}
	
	The discussion in the proof outline with \Cref{lem:oner_j} and \Cref{lem:intersectQ_j} proves the lemma. 
\end{proof}

\section{Applications} \label{sec:apps}

In this section, we give a more detailed exposition of the applications of \Cref{thm:paddedTW} and \Cref{thm:CoverTW} mentioned in \Cref{sec:intro}. The list of applications here is not meant to be exhaustive, and we believe that our result will find further applications. 

\subsection{Flow Sparsifier}
Given an edge-capacitated graph $G = (V,E,c)$ and a set $K\subseteq V$ of terminals, a $K$-flow is a flow where all the endpoints are terminals. A flow-sparsifier with quality $\rho\ge1$ is another capacitated graph $H = (K,E_H, c_H)$ such that (a) any feasible $K$-flow in $G$ can be feasibly routed in $H$, and (b) any feasible $K$-flow in $H$ can be routed in $G$ with congestion $\rho$ (see \cite{EGKRTT14} for formal definitions).

Englert \etal \cite{EGKRTT14} showed that given a graph $G$ which admits a $(\beta,\frac1\beta)$-padded decomposition scheme, then for any subset $K$, one can efficiently compute a flow-sparsifier with
quality $O(\beta)$. Using their result, we obtain the following corollary of \Cref{thm:paddedTW}. 

\begin{corollary}\label{cor:flowSpar}
	Given an edge-capacitated graph $G = (V,E,c)$ with treewidth $\tw$, and a subset of terminals $K\subseteq V$, one can efficiently compute a flow sparsifier with quality $O(\log\tw)$. 
\end{corollary}

The best previously known result had quality $O(\tw)$ approximation \cite{EGKRTT14,AGGNT19}. Thus, our result improves the dependency on $\tw$ exponentially.  For further reading on flow sparsifiers, see \cite{Moitra09,MT10,MM10,CLLM10,Chuzhoy12,AGK14}.

We note that, though the graph $G$ has treewidth $\tw$, the flow-sparsifier $H$ in \Cref{cor:flowSpar} can have arbitrarily large treewidth. Having low treewidth is a very desirable property of a graph, and naturally, we would like to have a sparsifier of small treewidth. A flow-sparsifier  $H = (K,E_H,  c_H)$ of $G$ is called \emph{minor-based} if $H$ is a minor of $G$. That is, $H$ can be obtained from $G$ by deleting/contracting edges, and deleting vertices. Englert \etal \cite{EGKRTT14} showed that given a graph $G$ which admits a $(\beta,\frac1\beta)$-padded decomposition scheme, then for any subset $K$, one can efficiently compute a minor-based flow-sparsifier with
quality $O(\beta\log\beta)$. By \Cref{thm:paddedTW}, we obtain:

\begin{corollary}\label{cor:flowSparMinorBased}
	Given an edge-capacitated graph $G = (V,E,c)$ with treewidth $\tw$, and a subset of terminals $K\subseteq V$, one can efficiently compute a minor-based flow-sparsifier  $H = (K,E_H, c_H)$ with quality $O(\log\tw\log\log\tw)$. In particular, $H$ also has treewidth at most $\tw$.
\end{corollary}

\subsection{Sparse Partition}

Given a metric space $(X,d_X)$,
a \emph{$(\alpha, \tau,\Delta)$-sparse partition} is a partition $\cC$ of $X$ such that:
\begin{itemize}
	\item \textbf{Low Diameter:} $\forall C\in\cC$, the set $X$ has diameter at most $\Delta$;
	\item \textbf{Sparsity:} $\forall x \in X$, the ball $B_X(x, \frac{\Delta}{\alpha})$ intersects at most $\tau$ clusters from $\cC$.
\end{itemize}
We say that the metric $(X,d_X)$ admits a $(\alpha, \tau)$-sparse partition scheme if for every $\Delta>0$, $X$ admits a $(\alpha, \tau,\Delta)$-sparse partition.

Jia \etal \cite{JLNRS05} implicitly proved (see \cite{Fil20} for an explicit proof) that if a space admits a $(\beta,s)$-sparse cover scheme, then it admits a $(\beta,s)$-sparse partition scheme. Therefore, by \Cref{thm:CoverTW} we obtain:

\begin{corollary}\label{cor:sparsePartition}
	Every graph $G$ with treewidth $\tw$ admits a $\left(O(1),\poly(\tw)\right)$-sparse partition scheme. 
\end{corollary}

Previously, it was only known that graphs with treewidth $\tw$ admit  $\left(O(\tw^2),2^{\tw}\right)$-sparse partition scheme \cite{FT03,Fil20}.  Thus, \Cref{cor:sparsePartition} implies an exponential improvement in the dependency on $\tw$.  For further reading on sparse partitions, see \cite{JLNRS05,BDRRS12,Fil20,CJFKVY23,Filtser24}.

\subsection{Universal Steiner Tree and Universal TSP}

We consider the problem of designing a network that allows a server to broadcast a message to a single set of clients. If sending a message over a link incurs some cost, then designing the best broadcast network is classically modeled as the Steiner tree problem \cite{HR92}. However, if the server has to solve this problem repeatedly with different client sets, it is desirable to construct a single network that will optimize the cost of the broadcast for every possible subset of clients. This setting motivates the \emph{Universal Steiner Tree} (\UST) problem.

Given a metric space $(X,d_X)$ and root  $r \in X$, a $\rho$-approximate \UST is a weighted tree $T$ over $X$ such that for every $S \subseteq X$ containing $r$, we have
\begin{align*}
	w(T\{S\}) \leq \rho \cdot \OPT_S
\end{align*}
where $T\{S\} \subseteq T$ is the minimal subtree of $T$ connecting $S$, and $\OPT_S$ is the minimum weight Steiner tree connecting $S$ in $X$.

A closely related problem to \UST is the \emph{Universal Traveling Salesman Problem} (\UTSP).
Consider a postman providing post service for a set $X$ of clients with $n$ different locations (with distance measure $d_X$). Each morning, the postman receives a subset $S\subset X$ of the required deliveries for the day. In order to minimize the total tour length, one solution may be to compute each morning an (approximation of an) Optimal \TSP tour for the set $S$. An alternative solution will be to compute a \emph{Universal \TSP} (\UTSP) tour $R$ containing all the points $X$. Given a subset $S$, $R\{S\}$ is the tour visiting all the points in $S$ w.r.t. the order induced by $R$.
Given a tour $T$, denote its length by $|T|$. The \emph{stretch} of $R$ is the maximum ratio among all subsets $S\subseteq X$ between the length of $R\{S\}$ and the length of the optimal \TSP tour on $S$, $\max_{S\subseteq X}\frac{|R\{S\}|}{|\mbox{Opt}(S)|}$.

Jia \etal \cite{JLNRS05} showed that for every $n$-point metric space that admits $(\sigma,\tau)$-sparse partition scheme, there is a polynomial time algorithm that given a root $\rt\in V$ computes a  \UST with stretch $O(\tau\sigma^2\log_\tau n)$.
In addition, Jia \etal \cite{JLNRS05} also showed that such a metric admit a \UTSP with stretch  $O(\tau\sigma^2\log_\tau n)$.
Using our \Cref{cor:sparsePartition}, we conclude:
\begin{corollary}\label{cor:UST}
	Consider an $n$-point graph with treewidth $\tw$. Then its shortest path metric admits a solution to both universal Steiner tree and universal TSP with stretch $\poly(\tw)\cdot\log n$. 
\end{corollary}
The best-known previous result for both problems had a stretch of $\exp(\tw)\cdot\log n$ \cite{Fil20}, which is exponentially larger than ours in terms of the $\tw$ dependency. For further reading on the \UTSP and \TSP problems, see \cite{platzman1989spacefilling,bertsimas1989worst,JLNRS05,GHR06,HKL06,schalekamp2008algorithms,gorodezky2010improved,bhalgat2011optimal,BDRRS12,BLT14,Fil20, Filtser24}. Interestingly, our solution to the \UST problem in \Cref{cor:UST} produces a solution which is not a subgraph of the low treewidth input graph $G$. If one requires that the \UST will be a subgraph of $G$, the current state of the are is by Busch \etal \cite{BCFHHR23} who obtained stretch $O(\log^7 n)$. To date, this is also the best known upper bound for graph with bounded treewidth.

\subsection{Steiner Point Removal}
Given a graph $G=(V,E,w)$, and a subset of terminals $K\subseteq V$, in the Steiner point removal problem, we are looking for a graph $H=(K,E_H,w_H)$, which is a minor of $G$. We say that $H$ has stretch $t$, if for every $u,v\in K$, $d_G(u,v)\le d_H(u,v)\le t\cdot d_G(u,v)$. 
Englert \etal \cite{EGKRTT14} showed that given a graph $G$ which admits a $(\beta,\frac1\beta)$-padded decomposition scheme, one can compute a distribution $\calD$ over minors $H=(K,E_H,w_H)$, such that for every $u,v\in K$, and $H\in\supp(\calD)$, $d_G(u,v)\le d_H(u,v)$, and $\E_{H\sim\calD}[d_H(u,v)]\le O(\beta\cdot\log\beta)\cdot d_G(u,v)$. That is, the minor has expected stretch $O(\beta \log(\beta))$. Thus, by \Cref{thm:paddedTW}, we obtain:

\begin{corollary}\label{cor:SPR}
	Given a weighted graph $G=(V,E,w)$ with treewidth $\tw$, and a subset of terminals $K\subseteq V$, one can efficiently sample a minor $H=(K,E_H,c_H)$ of $G$ such that for every $u,v\in K$, $d_G(u,v)\le d_H(u,v)$, and $\E[d_H(u,v)]\le O(\log\tw\log\log\tw)\cdot d_G(u,v)$.
\end{corollary}
The best previously known result had expected stretch $O(\tw\log\tw)$ \cite{EGKRTT14,AGGNT19}. 
For the classic Steiner point removal, where we are looking for a single minor, the stretch is $2^{O(\tw\log\tw)}$~\cite{CCLMST23}. For further reading on the Steiner point removal problem, see \cite{G01,CXKR06,KKN15,Che18,Fil19,FKT19,Fil20,HL22,CCLMST23}.

\subsection{Zero Extension}

In the \ZEX problem, the input is a set $X$, a terminal set $K\subseteq X$, 
a metric $d_K$ on $K$, and a cost function 
$c\colon \binom{X}{2} \to \mathbb{R}_{\ge 0}$.

The goal is to find a \emph{retraction} $f:X\rightarrow K$ 
that minimizes $\sum_{\{x,y\}\in \binom{X}{2}} c(x,y)\cdot d_K(f(x),f(y))$.

A retraction is a surjective function $f:X\rightarrow K$ that satisfies $f(x)=x$ for all $x\in K$. The \ZEX problem, first proposed by Karzanov~\cite{Kar98}, 
generalizes the \ProblemName{Multiway Cut} problem \cite{DJPSY92} 
by allowing $d_K$ to be any discrete metric (instead of a uniform metric).

Lee and Naor \cite{LN05} (see also \cite{AFHKTT04,CKR04}) showed that for the case where the metric $(K,d_K)$ on the terminals admits a $(\beta,\frac1\beta)$-padded-decomposition, there is an $O(\beta)$ upper bound. By \Cref{thm:paddedTW}, we get:

\begin{corollary}\label{cor:0Extension}
	Consider an instance of the \ZEX problem $\left(K\subseteq X,d,_k,c:\binom{X}{2}\rightarrow \mathbb{R}_+\right)$, where the metric $(K,d_K)$ is a sub-metric of a shortest path metric of a graph with treewidth $\tw$. Then, one can efficiently find a solution with cost at most $O(\log\tw)$ times the cost of the optimal.\\
	In particular, there is a $O(\log \tw)$-approximation algorithm for the multiway cut problem for graphs of treewidth $\tw$.	
\end{corollary}
The best previously known result was $O(\tw)$ approximation \cite{CKR04,AGGNT19}. For further reading on the \ZEX problem, see \cite{CKR04,FHRT03,AFHKTT04,LN05,EGKRTT14,FKT19}.

\subsection{Lipschitz Extension}

For a function $f:X\rightarrow Y$ between metric spaces $(X,d_X),(Y,d_Y)$, set $\|f\|_{\Lip}=\sup_{x,y\in X}\frac{d_Y(f(x),f(y))}{d_X(x,y)}$ to be the Lipschitz parameter of the function. In the Lipschitz extension problem, we are given a map $f:Z\rightarrow Y$ from a subset $Z$ of $X$. The goal is to extend $f$ to a function $\tilde{f}$ over the entire space $X$, while minimizing $\|\tilde{f}\|_{\Lip}$ as a function of $\|f\|_{\Lip}$.
Lee and Naor \cite{LN05}, proved that if a space admits a $(\beta,\frac1\beta)$-padded decomposition scheme, then given a function $f$ from a subset $Z\subseteq X$ into a closed convex set $C$ in some Banach space, one can extend $f$ into $\tilde{f}:X\rightarrow C$
such that $\|\tilde{f}\|_{\Lip}\le O(\beta)\cdot \|f\|_{\Lip}$.
By \Cref{thm:paddedTW}, we obtain:

\begin{corollary}[Lipschitz Extension]\label{thm:LipschitzExtension}
	Consider a graph $G=(V,E,w)$ with treewidth $\tw$, and let $f:V'\rightarrow C$ be a map from a subset $V'\subseteq V$ into $C$, where $C$ is a convex closed set of some Banach space.
	Then there is an extension $\tilde{f}:X\rightarrow C$ such that $\|\tilde{f}\|_{\Lip}\le O(\log\tw)\cdot \|f\|_{\Lip}$.
\end{corollary}

\subsection{Embedding into \texorpdfstring{$\ell_p$}{lp} spaces}
Metric embedding is a map between two metric spaces that preserves all pairwise distances up to a small stretch. We say that embedding $f:X\rightarrow Y$ between $(X,d_X)$ and $(Y,d_Y)$ has distortion $t$ if for every $x,y\in X$ it holds that $d_X(x,y)\le d_Y(f(x),f(y))\le t\cdot d_X(x,y)$. Krauthgamer, Lee, Mendel and Naor \cite{KLMN04} (improving over Rao \cite{Rao99}) showed that every $n$-point metric space that admits $(\beta,\frac1\beta)$-padded decomposition scheme can be embedded into an $\ell_p$ space with distortion $O(\beta^{1-\frac1p}\cdot(\log n)^{\frac1p})$.
Previously, it was known that the shortest path metric of an $n$ point graph with treewidth $\tw$ (or more generally $K_{\tw}$-minor free) embeds into $\ell_p$ space with distortion $O((\tw)^{1-\frac1p}\cdot(\log n)^{\frac1p})$ \cite{KLMN04,AGGNT19}.
By \Cref{thm:paddedTW}, we conclude:

\begin{corollary}\label{cor:EmbeddingLp}
	Let $G$ be an $n$-point weighted graph with treewidth $\tw$. Then there exist embedding of $G$ into $\ell_p$ ($1\le p\le\infty$) with distortion $O((\log \tw)^{1-\frac1p}\cdot(\log n)^{\frac1p})$.
\end{corollary}
Since every finite subset of the Euclidean space $\ell_2$ embed isometrically into $\ell_p$ (for $1\le p\le\infty$), it follows that for $1\le p\le\infty$ such $G$ can be embedded into $\ell_p$ space (in particular $\ell_1$) with distortion $O(\sqrt{\log \tw\cdot \log n})$.

A norm space of special interest is $\ell_\infty$, as every finite metric space embeds into $\ell_\infty$ isometrically (i.e. with distortion $1$, this is the so called Fr\'echet embedding). However, the dimension of such embedding is very large: $\Omega(n)$.
Krauthgamer \etal \cite{KLMN04} proved that every graph with treewidth $\tw$ (or more generally $K_{\tw}$-minor free) embeds into $\ell^d_\infty$ with dimension $d=\tilde{O}(3^\tw)\cdot\log n$ and distortion $O(\tw^2)$.
This was recently improved by Filtser \cite{Filtser24} who showed that every graph with treewidth $\tw$ (or more generally $K_{\tw}$-minor free) embeds into $\ell^d_\infty$ with dimension $d=\tilde{O}(\tw^2)\cdot\log n$  and distortion $O(\tw)$ (alternatively distortion $3+\eps$ and dimension $d=\tilde{O}(\frac{1}{\eps})^{\tw+1}\cdot\log n$).
More generally, Filtser \cite{Filtser24} proved that if a graph $G$ admits a $(\beta,s)$-padded partition cover scheme, then $G$ embeds into $\ell^d_\infty$ with distortion 
$(1+\eps)\cdot2\beta$ and dimension $d=O\left(\frac{s}{\eps}\cdot\log\frac{\beta}{\eps}\cdot\log(\frac{n\cdot\beta}{\eps})\right)$. Using our \Cref{thm:PaddedPartitionCoverTW}, we conclude:
\begin{corollary}\label{cor:EmbeddingLinfty}
	Let $G$ be an $n$-point weighted graph with treewidth $\tw$. Then there exist embedding of $G$ into $\ell^d_\infty$ with distortion $O(1)$ and dimension $d=\poly(\tw)\cdot\log n$.
\end{corollary}
Note that \Cref{cor:EmbeddingLinfty} provides an exponential improvement in the dependence on $\tw$ compared to the previous best known embedding into $\ell_\infty$ with constant distortion.
For further reading about metric embedding into $\ell_p$ spaces see \cite{OS81,Bou85,LLR95,Rao99,KLMN04,GNRS04,LR10,AFGN22,KLR19,Fil19,Kumar22,Filtser24}.

\subsection{Stochastic decomposition for minor-free graphs - a reduction from additive stretch embeddings}\label{subsec:stocdecomp}

A major open question is to determine the padding parameter of $ K_r$-minor-free graphs. The current state of the art is $O(r)$ \cite{AGGNT19} (see also \cite{Fil19Approx}), while the natural conjecture is exponentially smaller  $O(\log r)$. A somewhat weaker guarantee which we call $(t,p,\Delta)$-\emph{stochastic decomposition}, is a distribution over partitions with diameter $\Delta$ such that every pair of vertices at distance at most $\frac{\Delta}{t}$ is clustered together with probability at least $p$. Compared to padded decomposition, the guarantee here is over pairs (instead of balls), and there is a threshold distance for the guarantee (instead of a linear dependence as in padded decomposition). Nonetheless, for many applications, a stochastic decomposition is good enough. Moreover, in most cases, the parameters of padded and stochastic decompositions are the same. The only case that we are aware of where they are different is high-dimensional Euclidean spaces. In particular, for minor-free graphs, the best-known result is a  $(O(r),\frac12,\Delta)$-stochastic decomposition \cite{AGGNT19}. Here, we argue that if the parameters of recently studied stochastic embeddings with additive distortion are improved, then our \Cref{thm:paddedTW} will imply much better stochastic decompositions for minor-free graphs.

\begin{definition}[Stochastic embedding with additive distortion] Consider an $n$-vertex weighted graph $G=(V,E,w)$ with diameter $D$. A
	stochastic additive embedding of $G$ into a distribution over embeddings $f$ into graphs $H$ with treewidth $\tau$ and expected additive distortion $\eps\cdot D$, is a distribution $\calD$ over maps $f$ from $G$ into graphs $H$, such that
	\begin{itemize}
		\item Small treewidth: Any graph $H$ in the support has treewidth at most $\tau$.
		\item Dominating: for every $u, v\in V$, and $(f,H)\in\supp(\calD)$, $d_G(u, v)\le d_H(f(u),f(v))$.
		\item Expected additive distortion: for every $u, v\in V$, $\E_{(f,H)\in\supp(\calD)}\left[d_H(f(u),f(v))\right]\le d_G(u,v)+\eps\cdot D$.
	\end{itemize}  
\end{definition}
Recently, Filtser and Le \cite{FL22} (see also \cite{FKS19,CFKL20,FL21}) constructed a stochastic embedding of every $n$-vertex $K_r$-minor free graph $G$ with diameter $D$, into a distribution over graphs with treewidth $O(\frac{(\log\log n)^2}{\eps^2})\cdot \chi(r)$ and expected additive distortion $\eps\cdot D$. Here $\chi(r)$ is an extremely fast-growing function of $r$ (an outcome of the minor-structure theorem \cite{RS03}). 
If one would construct such stochastic embedding into treewidth $\poly(\frac r\eps)$, then \Cref{lem:AdditiveEmbeddingToDecomposition} bellow will imply stochastic decomposition with parameter $O(\log r)$. Furthermore, even allowing treewidth $\poly(\frac {r\cdot\log n}{\eps})$ will imply stochastic decomposition with parameter $O(\log r+\log\log n)$ which is yet unknown as well.
\begin{lemma}\label{lem:AdditiveEmbeddingToDecomposition}
	Suppose that there is a coordinate monotone function $\pi(\eps,r,n)$, such that every $n$-vertex $K_r$-minor free graph with diameter $D$ admits a stochastic embedding into a distribution over graphs with treewidth $\pi(\eps,r,n)$ and expected additive distortion $\eps\cdot D$.
	Then for every $\Delta>0$, every $n$-vertex $K_r$-minor free admits a $(t,\frac12,\Delta)$-stochastic decomposition, for $t=O(\log\pi(\Omega(\frac{1}{r}),r,n))$.
\end{lemma}
\begin{proof}
	Consider an $n$-vertex, $K_r$-minor-free graph $G=\left(V,E,w\right)$. 
	Fix $\Delta>0$. We will construct a distribution $\mathcal{D}$ over partitions of $G$ such that every $\mathcal{P}\in\supp(\mathcal{D})$ will be weakly $\Delta$-bounded and for every pair $u,v$ at distance $\frac{\Delta}{t}$, $u,v$ will be clustered together with probability at least $\frac12$.
	
	The partition will be constructed as follows:
	\begin{enumerate}
		\item Use Filtser's \cite{Fil19Approx} strong padded decomposition scheme with diameter parameter $\Delta'=O(\frac{r}{t})\cdot\Delta$ to sample a partition $\cP'$ of $V$. Note that every cluster $C\in\cP'$ is a $K_r$-minor free graph of diameter $\Delta'$.
		\item Fix $\eps=\frac{\Delta}{5\cdot t\cdot \Delta'}=\Omega(\frac1r)$. For every cluster $C\in\cP'$, sample an embedding $f_C$ of $G[C]$ into a graph $H_C$ with treewidth $\pi(\eps,r,|C|)\le \pi(\Omega(\frac1r),r,n)$ and expected additive distortion $\eps\cdot\Delta'$. We will abuse notation and use $v\in C$ to denote both $v$ and $f_C(v)$.
		\item For every $C\in \cP'$, use \Cref{thm:paddedTW} to sample a 
		$\left(O(\log( \pi(\Omega(\frac1r),r,n))),\Omega(1),\Delta\right)$-padded decomposition $\cP_C$ of $H_C$.
	\end{enumerate}
	The final partition we return is the union of all the created partitions $\cP=\cup_{C\in\cP'}(\cP_C\cap V)$. We claim that $\cP$ has a weak diameter at most $\Delta$. Consider a cluster $P\in\cP$ and $u,v\in P$. There is some cluster $C\in\cP'$ containing both $u$ and $v$. Then, it holds that $d_G(u,v)\le  d_{G[C]}(u,v)\le d_{H_C}(u,v)\le\Delta$.
	
	Consider a pair of vertices $u,v$ at distance $\le\frac\Delta t$.  Denote by $\Psi_1$ the event that the ball of radius $\frac\Delta t$
	around $v$ belongs to a single cluster in $\cP'$. Denote by $\Psi_2$ the event that conditioned on $\Psi_1$ occurring, $d_{H_C}(u,v)\le d_{G}(u,v)+5\cdot\eps\cdot\Delta'$. Denote by $\Psi_3$ the event that conditioned on $\Psi_2$ occurring, $u$ and $v$ belong to the same cluster of $\cP_C$.
	
	Let $C$ denote the cluster containing $v$ in $\cP'$. By choosing the constant in $\Delta'$ large enough, it holds that 
	\[
	\Pr\left[\Psi_1\right]=\Pr\left[\ball_{G}(v,\frac{\Delta}{t})\subseteq C\right]=\Pr\left[\ball_{G}(v,\frac{\Delta}{t\Delta'}\cdot\Delta')\subseteq C\right]\ge e^{-O(r)\cdot\frac{\Delta}{t\Delta'}}\ge\frac{4}{5}~.
	\]
	Note that if $\ball_{G}(v,\frac{\Delta}{t})\subseteq C$, that the entire shortest path from $u$ to $v$ is contained in $C$, and hence $d_{G[C]}(u,v)=d_{G}(u,v)$.
	Using Markov, the probability that the additive distortion by $f_C$ is too large is bounded by:
	\begin{align*}
		\Pr\left[\overline{\Psi_{2}}\mid\Psi_{1}\right] & =\Pr\left[d_{H_{C}}(u,v)> d_{G[C]}(u,v)+5\cdot\epsilon\cdot\Delta'\right]\\
		& \le\frac{\E\left[d_{H_{C}}(u,v)-d_{G[C]}(u,v)\right]}{5\cdot\epsilon\cdot\Delta'}\le\frac{1}{5}~.
	\end{align*}
	We conclude $\Pr\left[\Psi_{2}\mid\Psi_{1}\right]\ge\frac{4}{5}$.
	Note that if $\Psi_{2}$ indeed occurred, then  $d_{H_{C}}(u,v)\le d_{G[C]}(u,v)+5\cdot\epsilon\cdot\Delta'=d_{G[C]}(u,v)+\frac{\Delta}{t}\le \frac{2\Delta}{t}$.
	Finally in step 3 we sample a padded decomposition $C_\cP$ of $H_C$. 
	As $H_C$ has treewidth $\pi(\Omega(\frac1r),r,n)$, the probability that $u$ and $v$ are clustered together is bounded by:
	\[
	\Pr\left[\Psi_{3}\mid\Psi_{1}\wedge\Psi_{2}\right]\ge\Pr\left[\ball_{H_{C}}(v,\frac{2\Delta}{t})\subseteq\cP_{C}(c)\right]\ge e^{-O(\log\pi(\Omega(\frac{1}{r}),r,n))\cdot\frac{2}{t}}\ge\frac{4}{5}~,
	\]
	where the last inequality holds for a large enough constant in the definition of $t$.
	We conclude
	\[
	\Pr\left[\cP(v)=\cP(u)\right]\ge\Pr\left[\Psi_{1}\wedge\Psi_{2}\wedge\Psi_{3}\right]\ge\left(\frac{4}{5}\right)^{3}>\frac{1}{2}~.
	\]
\end{proof}

\printbibliography

\end{document}